%% file: main.tex
\documentclass[10pt]{article}
\pdfoutput=1

\usepackage[utf8]{inputenc}
\usepackage[T1]{fontenc}
\usepackage[english]{babel}
\usepackage{fullpage}
\usepackage{amsmath,amssymb,amsthm}
\usepackage{thmtools}
\usepackage{mathtools}
\usepackage{bbm}
\usepackage{comment}
\usepackage{booktabs}
\usepackage[textsize=scriptsize]{todonotes}
\usepackage[noend]{algpseudocode}
\usepackage{algorithm}
\usepackage{url}
\usepackage[hidelinks]{hyperref}

\input{definitions.tex}

\title{A Nearly Optimal and Agnostic Algorithm for \\Properly Learning a Mixture of $k$ Gaussians, for any Constant $k$}

\author{Jerry Li\thanks{Supported by NSF grant CCF-1217921 and DOE grant DE-SC0008923.}\\MIT\\\texttt{jerryzli@mit.edu}
\and                                                                            
Ludwig Schmidt\thanks{Supported by MADALGO and a grant from the MIT-Shell Energy Initiative.}\\MIT\\\texttt{ludwigs@mit.edu}}

\begin{document}
\maketitle

\input{abstract}

\newpage

\input{intro}

\input{paper_outline}

\input{notation}

\input{poly_program}

\input{clustering}

\input{general}

\section{Acknowledgements}
We thank Jayadev Acharya, Ilias Diakonikolas, Piotr Indyk, Gautam Kamath, Ankur Moitra, Rocco Servedio, and Ananda Theertha Suresh for helpful discussions.

\bibliographystyle{alpha}
\bibliography{refs}

\end{document}

%% file: definitions.tex
\declaretheorem{theorem}
\declaretheorem[sibling=theorem]{definition}
\declaretheorem[sibling=theorem]{lemma}
\declaretheorem[sibling=theorem]{corollary}
\declaretheorem[sibling=theorem]{fact}

\newcommand\numberthis{\addtocounter{equation}{1}\tag{\theequation}}

\DeclarePairedDelimiter{\norm}{\lVert}{\rVert}
\DeclarePairedDelimiter{\abs}{\lvert}{\rvert}

\newcommand{\eps}{\ensuremath{\epsilon}}

\newcommand{\OPT}{\ensuremath{\mathrm{OPT}}}

\newcommand{\poly}{\ensuremath{\mathrm{poly}}}

\newcommand{\Otilde}{\ensuremath{\widetilde{O}}}
\newcommand{\setC}{\ensuremath{\mathcal{C}}}
\newcommand{\setP}{\ensuremath{\mathcal{P}}}
\newcommand{\R}{\ensuremath{\mathbb{R}}}
\newcommand{\N}{\ensuremath{\mathbb{N}}}
\newcommand{\diff}{\ensuremath{\mathop{}\!\mathrm{d}}}
\newcommand{\setI}{\mathcal{I}}

\newcommand{\setM}{\mathcal{M}}
\newcommand{\familyI}{\mathfrak{I}}
\newcommand{\paramset}{\ensuremath{\Theta}}

\newcommand{\normalpdf}{\ensuremath{\mathcal{N}}}
\newcommand{\mixpdf}{\ensuremath{\mathcal{M}}}
\newcommand{\polystandardpdf}{\ensuremath{\mathcal{\widetilde{P}}}}
\newcommand{\polypdf}{\ensuremath{P}}
\newcommand{\pdens}{\ensuremath{p_\mathrm{dens}}}

\newcommand{\pdiff}{\ensuremath{{p_\theta^\mathrm{diff}}}}
\newcommand{\midp}{\ensuremath{\mathrm{mid}}}
\newcommand{\good}{\ensuremath{\mathrm{good}}}
\newcommand{\bad}{\ensuremath{\mathrm{bad}}}

\newcommand{\specialcell}[2][c]{%
  \begin{tabular}[#1]{@{}c@{}}#2\end{tabular}}

\newcommand{\Ak}{\ensuremath{{\mathcal{A}_K}}}
\newcommand{\hhat}{\widehat{h}}

\newcommand{\ed}{\stackrel{\text{def}}{=}}

\algnewcommand{\LineComment}[1]{\State \(\triangleright\) #1}
\makeatletter
\let\OldStatex\Statex
\renewcommand{\Statex}[1][3]{%
\setlength\@tempdima{\algorithmicindent}%
\OldStatex\hskip\dimexpr#1\@tempdima\relax}

%% file: abstract.tex
\begin{abstract}
Learning a Gaussian mixture model (GMM) is a fundamental problem in machine learning, learning theory, and statistics.
One notion of learning a GMM is proper learning: here, the goal is to find a mixture of $k$ Gaussians $\mixpdf$ that is close to the density $f$ of the unknown distribution from which we draw samples.
The distance between $\mixpdf$ and $f$ is typically measured in the total variation or $L_1$-norm.

We give an algorithm for learning a mixture of $k$ univariate Gaussians that is nearly optimal for any fixed $k$.
The sample complexity of our algorithm is $\Otilde(\frac{k}{\eps^2})$ and the running time is $(k \cdot \log\frac{1}{\eps})^{O(k^4)} + \Otilde(\frac{k}{\eps^2})$.
It is well-known that this sample complexity is optimal (up to logarithmic factors), and it was already achieved by prior work. 
However, the best known time complexity for proper learning a $k$-GMM was $\Otilde(\frac{1}{\eps^{3k-1}})$. 
In particular, the dependence between $\frac{1}{\eps}$ and $k$ was exponential.
We significantly improve this dependence by replacing the $\frac{1}{\epsilon}$ term with a $\log \frac{1}{\epsilon}$ while only increasing the exponent moderately.
Hence, for any fixed $k$, the $\Otilde (\frac{k}{\eps^2})$ term dominates our running time, and thus our algorithm runs in time which is \emph{nearly-linear} in the number of samples drawn.
Achieving a running time of $\text{poly}(k, \frac{1}{\eps})$ for proper learning of $k$-GMMs has recently been stated as an open problem by multiple researchers, and we make progress on this open problem.

Moreover, our approach offers an agnostic learning guarantee: our algorithm returns a good GMM even if the distribution we are sampling from is not a mixture of Gaussians.
To the best of our knowledge, our algorithm is the first agnostic proper learning algorithm for GMMs.
Again, the closely related question of agnostic and proper learning for GMMs in the high-dimensional setting has recently been raised as an open question, and our algorithm resolves this question in the univariate setting.

We achieve these results by approaching the proper learning problem from a new direction: we start with an accurate density estimate and then fit a mixture of Gaussians to this density estimate.
Hence, after the initial density estimation step, our algorithm solves an entirely deterministic optimization problem.
We reduce this optimization problem to a sequence of carefully constructed systems of polynomial inequalities, which we then solve with Renegar's algorithm.
Our techniques for encoding proper learning problems as systems of polynomial inequalities are general and can be applied to properly learn further classes of distributions besides GMMs.
\end{abstract}

%% file: intro.tex
\section{Introduction}
Gaussian mixture models (GMMs) are one of the most popular and important statistical models, both in theory and in practice.
A sample is drawn from a GMM by first selecting one of its $k$ components according to the mixing weights, and then drawing a sample from the corresponding Gaussian, each of which has its own set of parameters.
Since many phenomena encountered in practice give rise to approximately normal distributions, GMMs are often employed to model distributions composed of several distinct sub-populations.
GMMs have been studied in statistics since the seminal work of Pearson \cite{Pearson94} and are now used in many fields including astronomy, biology, and machine learning.
Hence the following are natural and important questions: (i) how can we efficiently ``learn'' a GMM when we only have access to samples from the distribution, and (ii) what rigorous guarantees can we give for our algorithms?

\subsection{Notions of learning}
\label{sec:notions}
There are several natural notions of learning a GMM, all of which have been studied in the learning theory community over the last 20 years.
The known sample and time complexity bounds differ widely for these related problems, and the corresponding algorithmic techniques are also considerably different (see Table \ref{table:overviewtable} for an overview and a comparison with our results).
In order of decreasing hardness, these notions of learning are:

\begin{description}
\item[Parameter learning] The goal in parameter learning is to recover the parameters of the unknown GMM (i.e., the means, variances, and mixing weights) up to some given additive error $\eps$.\footnote{Since the accuracy $\eps$ depends on the scale of the mixture (i.e., the variance), these guarantees are often specified relative to the variance.} 
\item[Proper learning] In proper learning, our goal is to find a GMM $M'$ such that the $L_1$-distance (or equivalently, the total variation distance) between our hypothesis $M'$ and the true unknown distribution is small.
\item[Improper learning / density estimation] Density estimation requires us to find \emph{any} hypothesis $\hhat$ such that the $L_1$ distance between $\hhat$ and the unknown distribution is small ($\hhat$ does not need to be a GMM).
\end{description}

\begin{table}[htb]
\begin{tabular}{ccccc}
\toprule
Problem type & \specialcell{Sample complexity\\lower bound} & \specialcell{Sample complexity\\upper bound} & \specialcell{Time complexity\\upper bound} & \specialcell{Agnostic\\guarantee} \\
\midrule
Parameter learning & & & & \\[.1cm]
$k = 2$ & $\Theta(\frac{1}{\eps^{12}})$ \cite{HP15} & $O(\frac{1}{\eps^{12}})$ \cite{HP15}& $O(\frac{1}{\eps^{12}})$ \cite{HP15} & no \\[.3cm]
general $k$ & $\Omega(\frac{1}{\eps^{6k-2}})$ \cite{HP15} & $O((\frac{1}{\eps})^{c_k})$ \cite{MV10} & $O((\frac{1}{\eps})^{c_k})$ \cite{MV10} & no \\
\midrule
Proper learning & & & & \\[.1cm]
$k = 2$ & $\Theta(\frac{1}{\eps^{2}})$ & $\Otilde(\frac{1}{\eps^{2}})$ \cite{DK14} & $\Otilde(\frac{1}{\eps^{5}})$ \cite{DK14} & no \\[.3cm]
general $k$ & $\Theta(\frac{k}{\eps^{2}})$ & $\Otilde(\frac{k}{\eps^2})$ \cite{AJOS14} & $\Otilde(\frac{1}{\eps^{3k - 1}})$ \cite{DK14,AJOS14} & no \\[.3cm]
\textbf{Our results}& & & & \\[.1cm]
$k = 2$ &  & $\Otilde(\frac{1}{\eps^{2}})$ & \mbox{\boldmath$\Otilde(\frac{1}{\eps^{2}})$} & \textbf{yes} \\[.3cm]
general $k$ &  & $\Otilde(\frac{k}{\eps^2})$ & \mbox{\boldmath$(k \log\frac{1}{\eps})^{O(k^4)} + \Otilde(\frac{k}{\eps^2})$} & \textbf{yes}  \\
\midrule
Density estimation & & & & \\[.1cm]
general $k$ & $\Theta(\frac{k}{\eps^{2}})$ & $\Otilde(\frac{k}{\eps^{2}})$ \cite{ADLS15} & $\Otilde(\frac{k}{\eps^{2}})$ \cite{ADLS15} & yes \\
\bottomrule
\end{tabular}
\caption{Overview of the best known results for learning a mixture of univariate Gaussians.
Our contributions (highlighted as bold) significantly improve on the previous results for proper learning: the time complexity of our algorithm is nearly optimal for any fixed $k$.
Moreover, our algorithm gives agnostic learning guarantees.
The constant $c_k$ in the time and sample complexity of \cite{MV10} depends only on $k$ and is at least $k$.
The sample complexity lower bounds for proper learning and density estimation are folklore results.
The only time complexity lower bounds known are the corresponding sample complexity lower bounds, so we omit an extra column for time complexity lower bounds.
}
\label{table:overviewtable}
\vspace{-.3cm}
\end{table}

Parameter learning is arguably the most desirable guarantee because it allows us to recover the unknown mixture parameters.
For instance, this is important when the parameters directly correspond to physical quantities that we wish to infer.
However, this power comes at a cost:
recent work on parameter learning has shown that $\Omega(\frac{1}{\eps^{12}})$ samples are already necessary to learn the parameters of a mixture of two univariate Gaussians with accuracy $\eps$ \cite{HP15} (note that this bound is tight, i.e., the paper also gives an algorithm with time and sample complexity $O(\frac{1}{\eps^{12}})$).
Moreover, the sample complexity of parameter learning scales exponentially with the number of components: for a mixture of $k$ univariate Gaussians, the above paper also gives a sample complexity lower bound of $\Omega(\frac{1}{\eps^{6k-2}})$.
Hence the sample complexity of parameter learning quickly becomes prohibitive, even for a mixture of two Gaussians and reasonable choices of $\eps$.

At the other end of the spectrum, improper learning has much smaller time and sample complexity.
Recent work shows that it is possible to estimate the density of a univariate GMM with $k$ components using only $\Otilde(\frac{k}{\eps^2})$ samples and time \cite{ADLS15}, which is tight up to logarithmic factors.
However, the output hypothesis produced by the corresponding algorithm is only a piecewise polynomial and not a GMM.
This is a disadvantage because GMMs are often desirable as a concise representation of the samples and for interpretability reasons.

Hence an attractive intermediate goal is proper learning: similar to parameter learning, we still produce a GMM as output.
On the other hand, we must only satisfy the weaker $L_1$-approximation guarantee between our hypothesis and the unknown GMM.
While somewhat weaker than parameter learning, proper learning still offers many desirable features: for instance, the representation as a GMM requires only $3k - 1$ parameters, which is significantly smaller than the at least $6 k (1 + 2 \log\frac{1}{\eps})$ many parameters produced by the piecewise polynomial density estimate of \cite{ADLS15} (note that the number of parameters in the piecewise polynomial also grows as the accuracy of the density estimate increases).
Moreover, the representation as a GMM allows us to provide simple closed-form expressions for quantities such as the mean and the moments of the learnt distribution, which are then easy to manipulate and understand.
In contrast, no such closed-form expressions exist when given a general density estimate as returned by \cite{ADLS15}.
Furthermore, producing a GMM as output hypothesis can be seen as a regularization step because the number of peaks in the density is bounded by $k$, and the density is guaranteed to be smooth.
This is usually an advantage over improper hypotheses such as piecewise polynomials that can have many more peaks or discontinuities, which makes the hypothesis harder to interpret and process.
Finally, even the most general parameter learning algorithms require assumptions on the GMMs such as identifiability, while our proper learning algorithm works for \emph{any} GMM.

Ideally, proper learning could combine the interpretability and conciseness of parameter learning with the small time and sample complexity of density estimation.
Indeed, recent work has shown that one can properly learn a mixture of $k$ univariate Gaussians from only $\Otilde(\frac{k}{\eps^2})$ samples \cite{DK14,AJOS14}, which is tight up to logarithmic factors.
However, the time complexity of proper learning is not yet well understood.
This is in contrast to parameter learning and density estimation, where we have strong lower bounds and nearly-optimal algorithms, respectively.
For the case of two mixture components, the algorithm of \cite{DK14} runs in time $\Otilde(\frac{1}{\eps^5})$.
However, the time complexity of this approach becomes very large for general $k$ and scales as $O(\frac{1}{\eps^{3k-1}})$ \cite{DK14, AJOS14}.
Note that this time complexity is much larger than the $\Otilde(\frac{k}{\eps^2})$ required for density estimation and resembles the exponential dependence between $\frac{1}{\eps}$ and $k$ in the $\Omega(\frac{1}{\eps^{6k-2}})$ lower bound for parameter learning.
Hence the true time complexity of properly learning a GMM is an important open question.
In particular, it is not known whether the exponential dependence between $\frac{1}{\eps}$ and $k$ is necessary.

In our work, we answer this question and show that such an exponential dependence between $\frac{1}{\eps}$ and $k$ can be avoided.
We give an algorithm with the same (nearly optimal) sample complexity as previous work, but our algorithm also runs in time which is \emph{nearly-optimal}, i.e. nearly-linear in the number of samples, for any fixed $k$.
It is worth noting that proper learning of $k$-GMMs in time $\text{poly}(k, \frac{1}{\eps})$ has been raised as an open problem \cite{Moitra14,Bigdata15}, and we make progress on this question.
Moreover, our learning algorithm is \emph{agnostic}, which means that the algorithm tolerates arbitrary amounts of worst-case noise in the distribution generating our samples.
This is an important robustness guarantee, which we now explain further.

\subsection{Robust learning guarantees}
\label{sec:introrobust}
All known algorithms for properly learning or parameter learning\footnote{Note that parameter learning is not well-defined if the samples do not come from a GMM.
Instead, existing parameter learning algorithms are not robust in the following sense: if the unknown distribution is not a GMM, the algorithms are not guaranteed to produce a set of parameters such that the corresponding GMM is close to the unknown density.} a GMM offer rigorous guarantees only when there is at most a very small amount of noise in the distribution generating our samples.
Typically, these algorithms work in a ``non-agnostic'' setting: they assume that samples come from a distribution that is exactly a GMM with probability density function (pdf) $f$, and produce a hypothesis pdf $\hhat$ such that for some given $\eps > 0$
\[
  \norm{f - \hhat}_1 \; \leq \; \eps \; .
\]
But in practice, we cannot typically expect that we draw samples from a distribution that truly corresponds to a GMM.
While many natural phenomena are well approximated by a GMM, such an approximation is rarely exact.
Instead, it is useful to think of such phenomena as GMMs corrupted by some amount of noise.
Hence it is important to design algorithms that still provide guarantees when the true unknown distribution is far from any mixture of $k$ Gaussians.

\paragraph{Agnostic learning} Therefore, we focus on the problem of \emph{agnostic} learning \cite{KSS94}, where our samples can come from \emph{any} distribution, not necessarily a GMM.
Let $f$ be the pdf of this unknown distribution and let $\setM_k$ be the set of pdfs corresponding to a mixture of $k$ Gaussians.
Then we define $\OPT_k$ to be the following quantity:
\[
  \OPT_k \; \ed \; \min_{h \in \setM_k} \norm{f - h}_1 \; ,
\]
which is the error achieved by the best approximation of $f$ with a $k$-GMM.
Note that this is a deterministic quantity, which can also be seen as the error incurred when projecting $f$ onto set $\setM_k$.
Using this definition, an agnostic learning algorithm produces a GMM with density $\hhat$ such that
\[
  \norm{f - \hhat}_1 \; \leq \; C \cdot \OPT_k + \eps
\]
for some given $\eps > 0$ and a universal constant\footnote{Strictly speaking, agnostic learning requires this constant $C$ to be 1. However, such a tight guarantee is impossible for some learning problems such as density estimation. Hence we allow any constant $C$ in an agnostic learning guarantee, which is sometimes also called \emph{semi-agnostic} learning.} $C$ that does not depend on $\eps$, $k$, or the unknown pdf $f$.

Clearly, agnostic learning guarantees are more desirable because they also apply when the distribution producing our samples does not match our model exactly (note also that agnostic learning is strictly more general than non-agnostic learning).
Moreover, the agnostic learning guarantee is ``stable'': when our model is close to the true distribution $f$, the error of the best approximation, i.e.,  $\OPT_k$, is small.
Hence an agnostic learning algorithm still produces a good hypothesis.

On the other hand, agnostic learning algorithms are harder to design because we cannot make any assumptions on the distribution producing our samples.
To the best of our knowledge, our algorithm is the first agnostic algorithm for properly learning a mixture of $k$ univariate Gaussians.
Note that agnostic learning has recently been raised as an open question in the setting of learning high-dimensional GMMs \cite{Vempala12}, and our agnostic univariate algorithm can be seen as progress on this problem.
Moreover, our algorithm achieves such an agnostic guarantee without any increase in time or sample complexity compared to the non-agnostic case.

\subsection{Our contributions}
We now outline the main contributions of this paper.
Similar to related work on proper learning \cite{DK14}, we restrict our attention to univariate GMMs.
Many algorithms for high-dimensional GMMs work via reductions to the univariate case, so it is important to understand this case in greater detail \cite{KMV10,MV10,HP15}.

First, we state our main theorem.
See Section \ref{sec:notation} for a formal definition of the notation.
The quantity $\OPT_k$ is the same as introduced above in Section \ref{sec:introrobust}.
\begin{theorem}
\label{thm:intro}
Let $f$ be the pdf of an arbitrary unknown distribution, let $k$ be a positive integer, and let $\eps > 0$.
Then there is an algorithm that draws $\Otilde(\frac{k}{\eps^2})$ samples from the unknown distribution and produces a mixture of $k$ Gaussians such that the corresponding pdf $\hhat$ satisfies
\[
  \norm{f - \hhat}_1 \; \leq \; 42 \cdot \OPT_k + \eps \; .
\]
Moreover, the algorithm runs in time
\[
  \left(k \cdot \log\frac{1}{\eps}\right)^{O(k^4)} \, + \, \Otilde\left(\frac{k}{\eps^2}\right) \; .
\]
\end{theorem}
We remark that we neither optimized the exponent $O(k^4)$, nor the constant in front of $\OPT_k$.
Instead, we see our result as a proof of concept that it is possible to agnostically and properly learn a mixture of Gaussians in time that is essentially fixed-parameter optimal.
As mentioned above, closely related questions about efficient and agnostic learning of GMMs have recently been posed as open problems, and we make progress on these questions.
In particular, our main theorem implies the following contributions:
\paragraph{Running time} The time complexity of our algorithm is significantly better than previous work on proper learning of GMMs.
For the special case of 2 mixture components studied in \cite{DK14} and \cite{HP15}, our running time simplifies to $\Otilde(\frac{1}{\eps^2})$.
This is a significant improvement over the $\Otilde(\frac{1}{\eps^5})$ bound in \cite{DK14}.
Moreover, our time complexity matches the best possible time complexity for density estimation of $2$-GMMs up to logarithmic factors.
This also implies that our time complexity is \emph{optimal} up to log-factors.

For proper learning of $k$ mixture components, prior work achieved a time complexity of $\Otilde(\frac{1}{\eps^{3k-1}})$ \cite{DK14,AJOS14}.
Compared to this result, our algorithm achieves an \emph{exponential} improvement in the dependence between $\frac{1}{\eps}$ and $k$: our running time contains only a $(\log\frac{1}{\eps})$ term raised to the $\poly(k)$-th power, not a $(\frac{1}{\eps})^k$.
In particular, the $\Otilde(\frac{k}{\eps^2})$ term in our running time dominates for any fixed $k$.
Hence the time complexity of our algorithm is \emph{nearly optimal for any fixed $k$}.

\paragraph{Agnostic learning} Our algorithm is the first proper learning algorithm for GMMs that is \emph{agnostic}.
Previous algorithms relied on specific properties of the normal distribution such as moments, while our techniques are more robust.
Practical algorithms should offer agnostic guarantees, and we hope that our approach is a step in this direction.
Moreover, it is worth noting that agnostic learning, i.e., learning under noise, is often significantly harder than non-agnostic learning.
One such example is learning parity with noise, which is conjectured to be computationally hard.
Hence it is an important question to understand which learning problems are tractable in the agnostic setting.
While the agnostic guarantee achieved by our algorithm is certainly not optimal, our algorithm still shows that it is possible to learn a mixture of Gaussians agnostically with only a very mild dependence on $\frac{1}{\eps}$.

\paragraph{From improper to proper learning} Our techniques offer a general scheme for converting \emph{improper} learning algorithms to \emph{proper} algorithms.
In particular, our approach applies to any parametric family of distributions that are well approximated by a piecewise polynomial in which the parameters appear polynomially and the breakpoints depend polynomially (or rationally) on the parameters.
As a result, we can convert purely approximation-theoretic results into proper learning algorithms for other classes of distributions, such as mixtures of Laplace or exponential distributions.
Conceptually, we show how to approach proper learning as a purely deterministic optimization problem once a good density estimate is available.
Hence our approach differs from essentially all previous proper learning algorithms, which use probabilistic arguments in order to learn a mixture of Gaussians.

\subsection{Techniques}
At its core, our algorithm fits a mixture of Gaussians to a density estimate.
In order to obtain an $\eps$-accurate and agnostic density estimate, we invoke recent work that has a time and sample complexity of $\Otilde(\frac{k}{\eps^2})$ \cite{ADLS15}.
The density estimate produced by their algorithm has the form of a piecewise polynomial with $O(k)$ pieces, each of which has degree $O(\log\frac{1}{\eps})$.
It is important to note that our algorithm does not draw any furthers samples after obtaining this density estimate --- the process of fitting a mixture of Gaussians is entirely deterministic.

Once we have obtained a good density estimate, the task of proper learning reduces to fitting a mixture of $k$ Gaussians to the density estimate.
We achieve this via a further reduction from fitting a GMM to solving a carefully designed system of polynomial inequalities.
We then solve the resulting system with Renegar's algorithm \cite{Renegar92a,Renegar92b}.
This reduction to a system of polynomial inequalities is our main technical contribution and relies on the following techniques.

\paragraph{Shape-restricted polynomials}
Ideally, one could directly fit a mixture of Gaussian pdfs to the density estimate.
However, this is a challenging task because the Gaussian pdf $\frac{1}{\sigma \sqrt{2\pi}} e^{-\frac{(x-\mu)^2}{2}}$ is not convex in the parameters $\mu$ and $\sigma$.
Thus fitting a mixture of Gaussians is a non-convex problem.

Instead of fitting mixtures of Gaussians directly, we instead use the notion of a \emph{shape restricted} polynomial.
We say that a polynomial is shape restricted if its coefficients are in a given semialgebraic set, i.e., a set defined by a finite number of polynomial equalities and inequalities.
It is well-known in approximation theory that a single Gaussian can be approximated by a piecewise polynomial consisting of three pieces with degree at most $O(\log\frac{1}{\eps})$ \cite{Timan63}.
So instead of fitting a mixture of $k$ Gaussian directly, we instead fit a mixture of $k$ shape-restricted piecewise polynomials.
By encoding that the shape-restricted polynomials must have the shape of Gaussian pdfs, we ensure that the mixture of shape-restricted piecewise polynomials found by the system of polynomial inequalities is close to a true mixture of $k$-Gaussians.
After we have solved the system of polynomial inequalities, it is easy to convert the shape-restricted polynomials back to a proper GMM.

\paragraph{\Ak-distance}
The system of polynomial inequalities we use for finding a good mixture of piecewise polynomials must encode that the mixture should be close to the density estimate.
In our final guarantee for proper learning, we are interested in an approximation guarantee in the $L_1$-norm.
However, directly encoding the $L_1$-norm in the system of polynomial inequalities is challenging because it requires knowledge of the intersections between the density estimate and the mixture of piecewise polynomials in order to compute the integral of their difference accurately.
Instead of directly minimizing the $L_1$-norm, we instead minimize the closely related \Ak-norm from VC (Vapnik–Chervonenkis) theory \cite{DL01}.
For functions with at most $K-1$ sign changes, the $\Ak$-norm exactly matches the $L_1$-norm.
Since two mixtures of $k$ Gaussians have at most $O(k)$ intersections, we have a good bound on the order of the $\Ak$-norm we use to replace the $L_1$-norm.
In contrast to the $L_1$-norm, we can encode the \Ak-norm without increasing the size of our system of polynomial inequalities significantly --- directly using the $L_1$-norm would lead to an exponential dependence on $\log\frac{1}{\eps}$ in our system of polynomial inequalities.

\paragraph{Adaptively rescaling the density estimate}
In order to use Renegar's algorithm for solving our system of polynomial inequalities, we require a bound on the accuracy necessary to find a good set of mixture components.
While Renegar's algorithm has a good dependence on the accuracy parameter, our goal is to give an algorithm for proper learning without \emph{any} assumptions on the GMM.
Therefore, we must be able to produce good GMMs even if the parameters of the unknown GMM are, e.g., doubly-exponential in $\frac{1}{\eps}$ or even larger.
Note that this issue arises in spite of the fact that our algorithm works in the real-RAM model: since different mixture parameters can have widely variying scales, specifying a single accuracy parameter for Renegar's algorithm is not sufficient.

We overcome this technical challenge by adaptively rescaling the parametrization used in our system of polynomial inequalities based on the lengths of the intervals $I_1, \ldots, I_s$ that define the piecewise polynomial density estimate $\pdens$.
Since $\pdens$ can only be large on intervals $I_i$ of small length, the best Gaussian fit to $\pdens$ can only have large parameters near such intervals. 
Hence, this serves as a simple way of identifying where we require more accuracy when computing the mixture parameters.

\paragraph{Putting things together}
Combining the ideas outlined above, we can fit a mixture of $k$ Gaussians with a carefully designed system of polynomial inequalities.
A crucial aspect of the system of polynomial inequalities is that the number of variables is $O(k)$, that the number of inequalities is $k^{O(k)}$, and the degree of the polynomials is bounded by $O(\log\frac{1}{\eps})$.
These bounds on the size of the system of polynomial inequalities then lead to the running time stated in Theorem \ref{thm:intro}.
In particular, the size of the system of polynomial inequalities is almost independent of the number of samples, and hence the running time required to solve the system scales only \emph{poly-logarithmically} with $\frac{1}{\eps}$.

\subsection{Related work}
Due to space constraints, it is impossible to summarize the entire body of work on learning GMMs here.
Therefore, we limit our attention to results with provable guarantees corresponding to the notions of learning outlined in Subsection \ref{sec:notions}.
Note that this is only one part of the picture: for instance, the well-known Expectation-Maximization (EM) algorithm is still the subject of current research (see \cite{BWY14} and references therein).

For parameter learning, the seminal work of Dasgupta \cite{Dasgupta99} started a long line of research in the theoretical computer science community, e.g., \cite{AK01,VW04,AM05,KSV08,BV08,KMV10}.
We refer the reader to \cite{MV10} for a discussion of these and related results.
The papers \cite{MV10} and \cite{BS10} were the first to give polynomial time algorithms (polynomial in $\eps$ and the dimension of the mixture) with provably minimal assumptions for $k$-GMMs.
More recently, Hardt and Price gave tight bounds for learning the parameters of a mixture of $2$ univariate Gaussians \cite{HP15}: $\Theta(\frac{1}{\eps^{12}})$ samples are necessary and sufficient,  and the time complexity is linear in the number of samples.
Moreover, Hardt and Price give a strong lower bound of $\Omega(\frac{1}{\eps^{6k-2}})$ for the sample complexity of parameter learning a $k$-GMM.
While our proper learning algorithm offers a weaker guarantee than these parameter learning approaches, our time complexity does not have an exponential dependence between $\frac{1}{\eps}$ and $k$.
Moreover, proper learning retains many of the attractive features of parameter learning (see Subsection \ref{sec:notions}).

Interestingly, parameter learning becomes more tractable as the number of dimensions increases.
A recent line of work investigates this phenomenon under a variety of assumptions (e.g., non-degeneracy or smoothed analysis) \cite{HK13,BCMV14,ABGRV14,GHK15}.
However, all of these algorithms require a lower bound on the dimension $d$ such as $d \geq \Omega(k)$ or $d \geq \Omega(k^2)$.
Since we focus on the one-dimensional case, our results are not directly comparable.
Moreover, to the best of our knowledge, none of the parameter learning algorithms (in any dimension) provide proper learning guarantees in the agnostic setting.

The first work to consider proper learning of $k$-GMMs without separation assumptions on the components was \cite{FSO06}.
Their algorithm takes $\poly (d, 1 / \epsilon, L)$ samples and returns a mixture whose KL-divergence to the unknown mixture is at most $\epsilon$.
Unfortunately, their algorithm has a pseudo-polynomial dependence on $L$, which is a bound on the means and the variances of the underlying components.
Note that such an assumption is not necessary a priori, and our algorithm works without any such requirements.
Moreover, their sample complexity is exponential in the number of components $k$.

The work closest to ours are the papers \cite{DK14} and \cite{AJOS14}, who also consider the problem of properly learning a $k$-GMM.
Their algorithms are based on constructing a set of candidate GMMs that are then compared via an improved version of the Scheff\'{e}-estimate.
While this approach leads to a nearly-optimal sample complexity of $\Otilde(\frac{k}{\eps^2})$, their algorithm constructs a large number of candidate hypothesis.
This leads to a time complexity of $O(\frac{1}{\eps^{3k-1}})$.
As pointed out in Subsection \ref{sec:notions}, our algorithm significantly improves the dependence between $\frac{1}{\eps}$ and $k$.
Moreover, none of their algorithms are agnostic.

Another related paper on learning GMMs is \cite{BSZ15}.
Their approach reduces the learning problem to finding a sparse solution to a non-negative linear system.
Conceptually, this approach is somewhat similar to ours in that they also fit a mixture of Gaussians to a set of density estimates.
However, their algorithm does not give a proper learning guarantee: instead of $k$ mixture components, the GMM returned by their algorithm contains $O(\frac{k}{\eps^3})$ components.
Note that this number of components is significantly larger than the $k$ components returned by our algorithm.
Moreover, their number of components increases as the accuracy paramter $\eps$ improves.
In the univariate case, the time and sample complexity of their algorithm is $O(\frac{k}{\eps^6})$.
Note that their sample complexity is not optimal and roughly $\frac{1}{\eps^4}$ worse than our approach.
For any fixed $k$, our running time is also better by roughly $\frac{1}{\eps^4}$.
Furthermore, the authors do not give an agnostic learning guarantee for their algorithm.

For density estimation, there is a recent line of work on improperly learning structured distributions \cite{CDSS13,CDSS14,ADLS15}.
While the most recent paper from this line achieves a nearly-optimal time and sample complexity for density estimation of $k$-GMMs, the hypothesis produced by their algorithm is a piecewise polynomial.
As mentioned in Subsection \ref{sec:notions}, GMMs have several advantages as output hypothesis.

%% file: paper_outline.tex
\subsection{Outline of our paper}
In Section \ref{sec:notation}, we introduce basic notation and important known results that we utilize in our algorithm.
Section \ref{sec:wbh} describes our learning algorithm for the special case of \emph{well-behaved} density estimates.
This assumption allows us to introduce two of our main tools (shape-restricted polynomials and the \Ak-distance as a proxy for $L_1$) without the technical details of adaptively reparametrizing the shape-restricted polynomials.
Section \ref{sec:generalalg} then removes this assumption and gives an algorithm that works for agnostically learning \emph{any} mixture of Gaussians.
In Section \ref{sec:general}, we show how our techniques can be extended to properly learn further classes of distributions.

%% file: notation.tex

\section{Preliminaries}
\label{sec:notation}
Before we construct our learning algorithm for GMMs, we introduce basic notation and the necessary tools from density estimation, systems of polynomial inequalities, and approximation theory.

\subsection{Basic notation and definitions}
For a positive integer $k$, we write $[k]$ for the set $\{1, \ldots, k\}$.
Let $I = [\alpha, \beta]$ be an interval.
Then we denote the length of $I$ with $|I| = \beta - \alpha$.
For a measurable function $f : \R \to \R$, the $L_1$-norm of $f$ is $\norm{f}_1 = \int f(x) \diff x$.
All functions in this paper are measurable.

Since we work with systems of polynomial inequalities, it will be convenient for us to parametrize the normal distribution with the \emph{precision}, i.e., one over the standard deviation, instead of the variance.
Thus, throughout the paper we let 
\[
\normalpdf_{\mu, \tau} (x) \; \ed \; \frac{\tau}{\sqrt{2 \pi}} \, e^{- \tau^2 (x - \mu)^2 / 2}
\]
denote the pdf of a normal distribution with mean $\mu$ and precision $\tau$. 
A $k$-GMM is a distribution with pdf of the form $\sum_{i = 1}^k w_i \cdot \normalpdf_{\mu_i, \tau_i} (x)$, where we call the $w_i$ \emph{mixing weights} and require that the $w_i$ satisfy $w_i \geq 0$ and $\sum_{i = 1}^k w_i = 1$.
Thus a $k$-GMM is parametrized by $3k$ parameters; namely, the mixing weights, means, and precisions of each component.\footnote{Note that there are only $3k - 1$ degrees of freedom since the mixing weights must sum to 1.}
We let $\paramset_k = \mathcal{S}_k \times \R^k \times \R_{+}^k$ be the set of parameters, where $\mathcal{S}_k$ is the simplex in $k$ dimensions.
For each $\theta \in \paramset_k$, we identify it canonically with $\theta = (w, \mu, \tau)$ where $w, \mu,$ and $\tau$ are each vectors of length $k$, and we let
\[
\mixpdf_\theta (x) = \sum_{i = 1}^k w_i \cdot \normalpdf_{\mu_i, \tau_i} (x)
\]
be the pdf of the $k$-GMM with parameters $\theta$.

\subsection{Important tools}
We now turn our attention to results from prior work.

\subsubsection{Density estimation with piecewise polynomials}
Our algorithm uses the following result about density estimation of $k$-GMMs as a subroutine.
\begin{fact}[\cite{ADLS15}]
\label{thm:ADLS}
Let $k \geq 1$, $\epsilon > 0$ and $\delta > 0$.
There is an algorithm $\textsc{Estimate-Density}(k, \epsilon, \delta)$ that satisfies the following properties:
the algorithm
\begin{itemize}
\item takes $\Otilde((k + \log(1 / \delta)) / \epsilon^2)$ samples from the unknown distribution with pdf $f$,
\item runs in time $\Otilde((k + \log 1 / \delta) / \epsilon^2)$, and
\item returns $\pdens$, an $O(k)$-piecewise polynomial of degree $O(\log (1 / \epsilon))$ such that
\[
  \| f - \pdens \|_1 \; \leq \; 4 \cdot \OPT_k + \epsilon
\]
with probability at least $1 - \delta$, where
\[
  \OPT_k \; = \; \min_{\theta \in \paramset_k} \norm{f - \mixpdf_\theta}_1 \; .
\]
\end{itemize}
\end{fact}

\subsubsection{Systems of polynomial inequalities}
In order to fit a $k$-GMM to the density estimate, we solve a carefully constructed system of polynomial inequalities.
Formally, a system of polynomial inequalities is an expression of the form
\[S = (Q_1 x^{(1)} \in \R^{n_1}) \ldots (Q_v x^{(v)} \in \R^{n_v}) P(y, x^{(1)}, \ldots, x^{(v)})\]
where
\begin{itemize}
\item
the $y = (y_1 \ldots, y_\ell)$ are free variables,
\item
for all $i \in [v]$, the quantifier $Q_i$ is either $\exists$ or $\forall$,
\item $P(y, x^{(1)}, \ldots, x^{(v)})$ is a quantifier-free Boolean formula with $m$ predicates of the form
\[g_i (y, x^{(1)}, \ldots, x^{(v)}) \; \Delta_i \; 0\]
where each $g_i$ is a real polynomial of degree $d$, and where the relations $\Delta_i$ are of the form $\Delta_i \in \{< , \geq, =, \neq, \leq, <\}$.
We call such predicates \emph{polynomial predicates}.
\end{itemize}

We say that $y \in \R^\ell$ is a \emph{$\lambda$-approximate solution} for this system of polynomial inequalities if there exists a $y' \in \R^\ell$ such that $y'$ satisfies the system and $\norm{y - y'}_2 \leq \lambda$.
We use the following result by Renegar as a black-box:

\begin{fact}[\cite{Renegar92a,Renegar92b}]
\label{thm:Renegar}
Let $0 < \lambda < \eta$ and let $S$ be a system of polynomial inequalities as defined above.
Then there is an algorithm $\textsc{Solve-Poly-System}(S, \lambda, \eta)$ that finds a $\lambda$-approximate solution if there exists a solution $y$ with $\norm{y}_2 \leq \eta$.
If no such solution exists, the algorithm returns ``NO-SOLUTION''.
In any case, the algorithm runs in time
\[
  (m d)^{2^{O(v)}  \ell \prod_k n_k} \log \log \left( 3 + \frac{\eta}{\lambda} \right) \; .
\]
\end{fact}

\subsubsection{Shape-restricted polynomials} 
Instead of fitting Gaussian pdfs to our density estimate directly, we work with piecewise polynomials as a proxy.
Hence we need a good approximation of the Gaussian pdf with a piecewise polynomial.
In order to achieve this, we use three pieces: two flat pieces that are constant 0 for the tails of the Gaussian, and a center piece that is given by the Taylor approximation.

Let let $T_d (x)$ be the degree-$d$ Taylor series approximation to $\normalpdf$ around zero. It is straightforward to show:
\begin{lemma}
\label{lem:taylor}
Let $\eps, K > 0$ and let $T_d (x)$ denote the degree-$d$ Taylor expansion of the Gaussian pdf \normalpdf{} around $0$. 
For $d = 2K \log (1 / \epsilon)$, we have
\[
\int_{2 \sqrt{\log 1 / \epsilon}}^{2 \sqrt{\log 1 / \epsilon}} |\normalpdf(x) - T_d (x)| \diff x \; \leq \; O\left( \epsilon^K \sqrt{\log (1 / \epsilon)} \right) \;.
\]
\end{lemma}

\begin{definition}[Shape-restricted polynomials]
Let $K$ be such that
\[
  \int_{-2 \sqrt{\log 1 / \epsilon}}^{2 \sqrt{\log 1 / \epsilon}} |\normalpdf(x) - T_{2 K \log(1 / \eps)}(x)| \diff x \; < \; \frac{\epsilon}{4} \; .
\]
From Lemma \ref{lem:taylor} we know that such a $K$ always exists.
For any $\eps > 0$, let $\polystandardpdf_{\epsilon} (x)$ denote the piecewise polynomial function defined as follows:
\[
  \polystandardpdf_{\epsilon}(x) \; = \; \begin{cases} T_{2 K \log(1 / \eps)}(x)
  \quad & \text{if } x \in [-2 \sqrt{\log (1 / \epsilon)}, \, 2 \sqrt{\log (1 / \epsilon)}] \\ 0 & \text{otherwise} \end{cases} \; .
\]

For any set of parameters $\theta \in \paramset_k$, let 
\[
  \polypdf_{\epsilon, \theta} (x) \; = \; \sum_{i = 1}^k w_i \cdot \tau_i \cdot \polystandardpdf_\epsilon (\tau_i (x - \mu_i)) \; .
\]
\end{definition}

It is important to note that $\polypdf_{\epsilon, \theta} (x)$ is a polynomial \emph{both} as a function of $\theta$ and as a function of $x$.
This allows us to fit such shape-restricted polynomials with a system of polynomial inequalities.
Moreover, our shape-restricted polynomials are good approximations to GMMs.
By construction, we get the following result:
\begin{lemma}
\label{lem:goodapprox}
Let $\theta \in \paramset_k$. Then $\norm{\mixpdf_\theta - \polypdf_{\epsilon, \theta}}_1 \leq \epsilon$.
\end{lemma}
\begin{proof}
We have
\begin{align*}
\| \mixpdf_\theta - \polypdf_{\epsilon, \theta} \|_1  &= \int | \mixpdf_\theta (x) - \polypdf_{\epsilon, \theta} (x) | \diff x \\
&\stackrel{(a)}{\leq} \sum_{i = 1}^k w_i \int | \tau_i \cdot \normalpdf(\tau_i (x - \mu_i)) - \tau_i \cdot \polystandardpdf_{\epsilon} (\tau_i (x - \mu_i)) | \diff x \\
&\stackrel{(b)}{\leq} \sum_{i = 1}^k w_i \cdot \| \normalpdf - \polystandardpdf_\eps \|_1 \\
&\stackrel{(c)}{\leq} \sum_{i = 1}^k w_i \cdot \eps\\
& \leq \epsilon \; .
\end{align*}
Here, (a) follows from the triangle inequality, (b) from a change of variables, and (c) from the definition of $\polystandardpdf_\eps$.
\end{proof}

\subsubsection{\texorpdfstring{$\Ak$}{AK}-norm and intersections of \texorpdfstring{$k$}{k}-GMMs}
In our system of polynomial inequalities, we must encode the constraint that the shape-restricted polynomials are a good fit to the density estimate.
For this, the following notion of distance between two densities will become useful.
\begin{definition}[$\Ak$-norm] Let $\familyI_K$ denote the family of all sets of $K$ disjoint intervals $\setI = \{ I_1, \ldots, I_K \}$.
For any measurable function $f: \R \to \R$, we define the \emph{$\Ak$-norm} of $f$ to be
\[\| f \|_{\Ak} \ed \sup_{\setI \in \familyI_K} \sum_{I \in \setI} \left| \int_I f(x) \diff x \right| \; .\]
\end{definition}
For functions with few zero-crossings, the \Ak-norm is close to the $L_1$-norm.
More formally, we have the following properties, which are easy to check:
\begin{lemma}
\label{fact:akproperties}
Let $f : \R \to \R$ be a real function.
Then for any $K \geq 1$, we have
\[
  \norm{f}_\Ak \; \leq \; \norm{f}_1 \; .
\]
Moreover, if $f$ is continuous and there are at most $K - 1$ distinct values $x$ for which $f(x) = 0$, then
\[
  \norm{f}_\Ak \; = \; \norm{f}_1 \; .
\]
\end{lemma}
The second property makes the $\Ak$-norm useful for us because linear combinations of Gaussians have few zeros.
\begin{fact}[\cite{KMV10} Proposition 7]
\label{fact:gmmintersection}
Let $f$ be a linear combination of $k$ Gaussian pdfs with variances $\sigma_1, \ldots, \sigma_k$ so that $\sigma_i \neq \sigma_j$ for all $i \neq j$.
Then there are at most $2(k - 1)$ distinct values $x$ such that $f(x) = 0$. 
\end{fact}
These facts give the following corollary.
\begin{corollary}
\label{cor:gmmak}
Let $\theta_1, \theta_2 \in \paramset_k$ and let $K \geq 4k$.
Then
\[
  \norm{\mixpdf_{\theta_1} - \mixpdf_{\theta_2}}_\Ak \; = \; \norm{\mixpdf_{\theta_1} - \mixpdf_{\theta_2}}_1 \; .
\]
\end{corollary}
\begin{proof}
For any $\gamma > 0$, let $\theta_1^\gamma, \theta_2^\gamma$ be so that $\| \theta_i^\gamma - \theta_i \|_\infty \leq \gamma$ for $i \in \{1, 2\}$,
and so that the variances of all the components in $\theta_1^\gamma, \theta_2^\gamma$ are all distinct.
Lemma \ref{fact:akproperties} and Fact \ref{fact:gmmintersection} together imply that $\|M_{\theta_1^\gamma} - M_{\theta_2^\gamma} \|_1 = \|M_{\theta_1^\gamma} - M_{\theta_2^\gamma} \|_\Ak$.
Letting $\gamma \to 0$ the LHS tends to $\norm{\mixpdf_{\theta_1} - \mixpdf_{\theta_2}}_\Ak$, and the RHS tends to $\norm{\mixpdf_{\theta_1} - \mixpdf_{\theta_2}}_1$.
So we get that $\norm{\mixpdf_{\theta_1} - \mixpdf_{\theta_2}}_\Ak \; = \; \norm{\mixpdf_{\theta_1} - \mixpdf_{\theta_2}}_1$, as claimed.
\end{proof}

%% file: poly_program.tex

\section{Proper learning in the well-behaved case}
\label{sec:wbh}
In this section, we focus on properly learning a mixture of $k$ Gaussians under the assumption that we have a ``well-behaved'' density estimate.
We study this case first in order to illustrate our use of shape-restricted polynomials and the \Ak-norm.
Intuitively, our notion of ``well-behavedness'' requires that there is a good GMM fit to the density estimate such that the mixture components and the overall mixture distribution live at roughly the same scale.
Algorithmically, this allows us to solve our system of polynomial inequalities with sufficient accuracy.
In Section \ref{sec:generalalg}, we remove this assumption and show that our algorithm works for \emph{all} univariate mixtures of Gaussians and requires no special assumptions on the density estimation algorithm.

\subsection{Overview of the Algorithm}
The first step of our algorithm is to learn a good piecewise-polynomial approximation $\pdens$ for the unknown density $f$.
We achieve this by invoking recent work on density estimation \cite{ADLS15}.
Once we have obtained a good density estimate, it suffices to solve the following optimization problem:
\[
\min_{\theta \in \paramset_k} \| \pdens - \mixpdf_{\theta} \|_1 \; .
\]
Instead of directly fitting a mixture of Gaussians, we use a mixture of shape-restricted piecewise polynomials as a proxy and solve
\begin{align*}
\min_{\theta \in \paramset_k} \norm{ \pdens - \polypdf_{\eps, \theta}}_1 \; .
\end{align*}
Now all parts of the optimization problem are piecewise polynomials.
However, we will see that we cannot directly work with the $L_1$-norm without increasing the size of the corresponding system of polynomial inequalities substantially.
Hence we work with the $\Ak$-norm instead and solve
\begin{align*}
\min_{\theta \in \paramset_k} \norm{ \pdens - \polypdf_{\eps, \theta}}_\Ak \; .
\end{align*}
We approach this problem by converting it to a system of polynomial inequalities with 
\begin{enumerate}
\item $O(k)$ free variables: one per component weight, mean, and precision, 
\item Two levels of quantification: one for the intervals of the \Ak-norm, and one for the breakpoints of the shape-restricted polynomial.
Each level quantifies over $O(k)$ variables.
\item A Boolean expression on polynomials with $k^{O(k)}$ many constraints.
\end{enumerate}
Finally, we use Renegar's algorithm to approximately solve our system in time $(k \log 1 / \epsilon)^{O(k^4)}$.
Because we only have to consider the well-behaved case, we know that finding a polynomially good approximation to the parameters will yield a sufficiently close approximation to the true underlying distribution.

\subsection{Density estimation, rescaling, and well-behavedness}
\paragraph{Density estimation}
As the first step of our algorithm, we obtain an agnostic estimate of the unknown probability density $f$.
For this, we run the density estimation subroutine $\textsc{Estimate-Density}(k, \epsilon, \delta)$ from Fact \ref{thm:ADLS}.
Let $\pdens'$ be the resulting $O(k)$-piecewise polynomial.
In the following, we condition on the event that
\[
  \| f - \pdens' \|_1 \; \leq \; 4 \cdot \OPT_k + \epsilon \; .
\]
which occurs with probability $1 - \delta$.

\paragraph{Rescaling}
Since we can solve systems of polynomial inequalities only with bounded precision,  we have to post-process the density estimate.
For example, it could be the case that some mixture components have extremely large mean parameters $\mu_i$, in which case accurately approximating these parameters could take an arbitrary amount of time.
Therefore, we shift and rescale $\pdens'$ so that its non-zero part is in $[-1, 1]$ (note that $\pdens$ can only have finite support because it consists of a bounded number of pieces).

Let $\pdens$ be the scaled and shifted piecewise polynomial.
Since the $L_1$-norm is invariant under shifting and scaling, it suffices to solve the following problem
\[
\min_{\theta \in \paramset_k} \| \pdens - \mixpdf_{\theta} \|_1 \; .
\]
Once we have solved this problem and found a corresponding  $\theta$ with
\[
  \norm{\pdens - \mixpdf_{\theta'}}_1 \; \leq \; C
\]
for some $C \geq 0$, we can undo the transformation applied to the density estimate and get a $\theta' \in \paramset_{k}$ such that
\[
  \norm{\pdens' - \mixpdf_{\theta'}}_1 \; \leq \; C \; .
\]

\paragraph{Well-behavedness}
While rescaling the density estimate $\pdens'$ to the interval $[-1, 1]$ controls the size of the mean parameters $\mu_i$, the precision parameters $\tau_i$ can still be arbitrarily large.
Note that for a mixture component with very large precision, we also have to approximate the corresponding $\mu_i$ very accurately.
For clarity of presentation, we ignore this issue in this section and assume that the density estimate is \emph{well-behaved}.
This assumption allows us to control the accuracy in Renegar's algorithm appropriately.
We revisit this point in Section \ref{sec:generalalg} and show how to overcome this limitation.
Formally, we introduce the following assumption:

\begin{definition}[Well-behaved density estimate]
\label{def:wellbehaved}
Let $\pdens'$ be a density estimate and let $\pdens$ be the rescaled version that is supported on the interval $[-1, 1]$ only.
Then we say $\pdens$ is $\gamma$-well-behaved if there is a set of GMM parameters $\theta \in \paramset_k$ such that
\[
  \norm{\pdens - \mixpdf_{\theta}}_1 \; = \; \min_{\theta^* \in \paramset_k} \norm{\pdens - \mixpdf_{\theta^*}}_1
\]
and $\tau_i \leq \gamma$ for all $i \in [k]$.
\end{definition}

The well-behaved case is interesting in its own right because components with very high precision parameter, i.e., very spiky Gaussians, can often be learnt by clustering the samples.\footnote{However, very spiky Gaussians can still be very close, which makes this approach challenging in some cases -- see Section \ref{sec:generalalg} for details.}
Moreover, the well-behaved case illustrates our use of shape-restricted polynomials and the $\Ak$-distance without additional technical difficulties.

\subsection{The \texorpdfstring{\Ak}{AK}-norm as a proxy for the \texorpdfstring{$L_1$}{L1}-norm}
Computing the $L_1$-distance between the density estimate $\pdens$ and our shape-restricted polynomial approximation $\polypdf_{\eps, \theta}$ \emph{exactly} requires knowledge of the zeros of the piecewise polynomial $\pdens - \polypdf_{\epsilon, \theta}$.
In a system of polynomial inequalities, these zeros can be encoded by introducing auxiliary variables.
However, note that we cannot simply introduce one variable per zero-crossing without affecting the running time significantly: since the polynomials have degree $O(\log 1/\eps)$, this would lead to $O(k \log 1 / \eps)$ variables, and hence the running time of Renegar's algorithm would depend exponentially on $O(\log 1/\eps)$.
Such an exponential dependence on $\log (1/\eps)$ means that the running time of solving the system of polynomial inequalities becomes super-polynomial in $\frac{1}{\eps}$, while our goal was to avoid any polynomial dependence on $\frac{1}{\eps}$ when solving the system of polynomial inequalities.

Instead, we use the $\Ak$-norm as an \emph{approximation} of the $L_1$-norm.
Since both $\polypdf_{\epsilon, \theta}$ and $\pdens$ are close to mixtures of $k$ Gaussians, their difference only has $O(k)$ zero crossings that contribute significantly to the $L_1$-norm.
More formally, we should have $\| \pdens - \polypdf_{\epsilon, \theta}\|_{1} \approx \| \pdens - \polypdf_{\epsilon, \theta} \|_{\Ak}$.
And indeed:
\begin{lemma}
\label{lem:Akapprox}
Let $\eps > 0$, $k \geq 2$, $\theta \in \paramset_k$, and $K = 4k$.
Then we have 
\[
  0 \; \leq \; \| \pdens - \polypdf_{\epsilon, \theta} \|_{1} - \| \pdens - \polypdf_{\epsilon, \theta} \|_\Ak \; \leq \; 8 \cdot \OPT_k + O(\epsilon) \; .\]
\end{lemma}
\begin{proof}
Recall Lemma \ref{fact:akproperties}: for any function $f$, we have $\| f \|_\Ak \leq \| f \|_{1}$.
Thus, we know that $ \| \pdens - \polypdf_{\epsilon, \theta} \|_\Ak \leq \| \pdens - \polypdf_{\epsilon, \theta} \|_{1}$.
Hence, it suffices to show that $\| \pdens - \polypdf_{\epsilon, \theta} \|_{1} \leq  8 \cdot \OPT_k + O(\epsilon) + \| \pdens - \polypdf_{\epsilon, \theta} \|_\Ak$.

We have conditioned on the event that the density estimation algorithm succeeds.
So from Fact \ref{thm:ADLS}, we know that there is some mixture of $k$ Gaussians $M_{\theta'}$ so that $\| \pdens - \mixpdf_{\theta'} \|_{1} \leq 4 \cdot \OPT_k + \epsilon$.
By repeated applications of the triangle inequality and Corollary \ref{cor:gmmak}, we get
\begin{align*}
\| \pdens - \polypdf_{\epsilon, \theta} \|_1 \; & \leq \; \| \pdens - \mixpdf_{\theta'} \|_1 + \| \mixpdf_{\theta'} - \mixpdf_\theta \|_1 + \| \polypdf_{\epsilon, \theta} - \mixpdf_\theta \|_{1} \; \\
					  & \leq \; 4 \cdot \OPT + \eps +  \| \mixpdf_{\theta'} - \mixpdf_{\theta} \|_{\Ak} + \eps \\
					  & \leq \; 4 \cdot \OPT + 2 \eps + \| \mixpdf_{\theta'} - \pdens \|_\Ak + \| \pdens - \polypdf_{\epsilon, \theta} \|_{\Ak} + \| \polypdf_{\epsilon, \theta} - \mixpdf_\theta \|_\Ak \\
					  & \leq \; 4 \cdot \OPT + 2 \eps + \| \mixpdf_{\theta'} - \pdens \|_1 + \| \pdens - \polypdf_{\epsilon, \theta} \|_{\Ak} + \| \polypdf_{\epsilon, \theta} - \mixpdf_\theta \|_1 \\
					  & \leq \; 8 \cdot \OPT + 4 \eps + \| \pdens - \polypdf_{\epsilon, \theta} \|_\Ak \;, \\
\end{align*}
as claimed.
\end{proof}

Using this connection between the $\Ak$-norm and the $L_1$-norm, we can focus our attention on the following problem:
\[
\min_{\theta \in \paramset_k} \| \pdens - \polypdf_{\epsilon, \theta} \|_{\Ak} \; .
\]
As mentioned above, this problem is simpler from a computational perspective because we only have to introduce $O(k)$ variables into the system of polynomial inequalities, regardless of the value of $\eps$.

When encoding the above minimization problem in a system of polynomial inequalities, we convert it to a sequence of feasibility problems.
In particular, we solve $O (\log (1 / \epsilon))$ feasibility problems of the form
\begin{equation}
\label{eq:wellbehavedgoal}
\text{Find } \; \theta \in \paramset_k \;\text{ s.t. } \| \pdens - \polypdf_{\epsilon, \theta} \|_\Ak < \nu \; .
\end{equation}
Next, we show how to encode such an \Ak-constraint in a system of polynomial inequalities.

\subsection{A general system of polynomial inequalities for encoding closeness in \texorpdfstring{\Ak}{AK}-norm}
\label{subsec:polyprogram}
In this section, we give a general construction for the \Ak-distance between any fixed piecewise polynomial (in particular, the density estimate) and any piecewise polynomial we optimize over (in particular, our shape-restricted polynomials which we wish to fit to the density estimate).
The only restriction we require is that we already have variables for the breakpoints of the polynomial we optimize over.
As long as these breakpoints depend only polynomially or rationally on the parameters of the shape-restricted polynomial, this is easy to achieve.
Presenting our construction of the \Ak-constraints in this generality makes it easy to adapt our techniques to the general algorithm (without the well-behavedness assumption, see Section \ref{sec:generalalg}) and to new classes of distributions (see Section \ref{sec:general}).

The setup in this section will be as follows.
Let $p$ be a given, fixed piecewise polynomial supported on $[-1, 1]$ with breakpoints  $c_1, \ldots, c_r$. 
Let $\setP$ be a set of piecewise polynomials so that for all $\theta \in S \subseteq \R^{u}$ for some fixed, known $S$, there is a $\polypdf_\theta (x) \in \setP$ with breakpoints $d_1 (\theta), \ldots, d_s (\theta)$ 
such that
\begin{itemize}
\item
$S$ is a semi-algebraic set.\footnote{Recall a semi-algebraic set is a set where membership in the set can be described by polynomial inequalities.}
Moreover, assume membership in $S$ can be stated as a Boolean formula over $R$ polynomial predicates, each of degree at most $D_1$, for some $R, D_1$.

\item
For all $1 \leq i \leq s$, there is a polynomial $h_i$ so that $h_i(d_i(\theta), \theta) = 0$, and moreover, for all $\theta$, we have that $d_i (\theta)$ is the unique real number $y$ satisfying $h_i (y, \theta) = 0$. That is, the breakpoints of $P_\theta$ can be encoded as polynomial equality in the $\theta$'s.
Let $D_2$ be the maximum degree of any $h_i$.

\item
The function $(x, \theta) \mapsto P_\theta (x)$ is a polynomial in $x$ and $\theta$ as long as $x$ is not at a breakpoint of $P_\theta$.
Let $D_3$ be the maximum degree of this polynomial.
\end{itemize}
Let $D = \max (D_1, D_2, D_3)$.

Our goal then is to encode the following problem as a system of polynomial inequalities:
\begin{equation}
\label{eq:polyprogramgoal}
\text{Find } \; \theta \in S \;\text{ s.t. } \| p - \polypdf_{\theta} \|_\Ak < \nu \; .
\end{equation}
In Section \ref{sec:polyprogramGMM}, we show that this is indeed a generalization of the problem in Equation (\ref{eq:wellbehavedgoal}), for suitable choices of $S$ and $\setP$.

In the following, let $\pdiff \ed p - \polypdf_{\theta}$.
Note that $\pdiff$ is a piecewise polynomial with breakpoints contained in $\{c_1, \ldots c_r, d_1 (\theta), \ldots, d_s (\theta) \}$.
In order to encode the $\Ak$-constraint, we use the fact that a system of polynomial inequalities can contain for-all quantifiers.
Hence it suffices to encode the \Ak-constraint for a single set of $K$ intervals.
We provide a construction for a single $\Ak$-constraint in Section \ref{subsubsec:Akfixed}.
In Section \ref{sec:wbhpolyprogram}, we introduce two further constraints that guarantee validity of the parameters $\theta$ and combine these constraints with the $\Ak$-constraint to produce the full system of polynomial inequalities.

\subsubsection{Encoding Closeness for a Fixed Set of Intervals}
\label{subsubsec:Akfixed}
Let $[a_1, b_1], \ldots, [a_K, b_K]$ be $K$ disjoint intervals.
In this section we show how to encode the following constraint:
\[
  \sum_{i=1}^K \left| \int_{a_i}^{b_i} \pdiff(x) \diff x \right| \; \leq  \; \nu \; .
\]
Note that a given interval $[a_i, b_i]$ might contain several pieces of \pdiff.
In order to encode the integral over $[a_i, b_i]$ correctly, we must therefore know the current order of the breakpoints (which can depend on $\theta$).

However, once the order of the breakpoints of \pdiff{} and the $a_i$ and $b_i$ is fixed, the integral over $[a_i, b_i]$ becomes the integral over a fixed set of sub-intervals.
Since the integral over a single polynomial piece is still a polynomial, we can then encode this integral over $[a_i, b_i]$ piece-by-piece.

More formally, let $\Phi$ be the set of permutations of the variables 
\[
  \{a_1, \ldots, a_K, \, b_1, \ldots, b_K, \, c_1, \ldots, c_r, \, d_1 (\theta), \ldots, d_s (\theta) \}
\]
such that (i) the $a_i$ appear in order, (ii) the $b_i$ appear in order, (iii) $a_i$ appears before $b_i$, and (iv) the $c_i$ appear in order.
Let $t = 2 K + r + s$.
For any $\phi = (\phi_1, \ldots, \phi_t) \in \Phi$, let
\[
  \text{ordered}_{p, \setP} (\phi) \ed \bigwedge_{i=1}^{t - 1} (\phi_i \leq \phi_{i + 1}) \; .
\]
Note that for any fixed $\phi$, this is an unquantified Boolean formula with polynomial constraints in the unknown variables.
The order constraints encode whether the current set of variables corresponds to ordered variables under the permutation represented by $\phi$.
An important property of an ordered $\phi$ is the following: in each interval $[\phi_i, \phi_{i+1}]$, the piecewise polynomial $\pdiff$ has exactly one piece.
This allows us to integrate over $\pdiff$ in our system of polynomial inequalities.

Next, we need to encode whether a fixed interval between $\phi_i$ and $\phi_{i+1}$ is contained in one of the $\Ak$-intervals, i.e., whether we have to integrate $\pdiff$ over the interval $[\phi_i, \phi_{i+1}]$ when we compute the $\Ak$-norm of \pdiff.
We use the following expression:
\[
  \text{is-active}_{p, \setP}(\phi, i) \ed \begin{cases}1 \quad & \begin{aligned}\text{if there is a $j$ such that $a_j$ appears as or before $\phi_i$ in $\phi$} \\ \text{and $b_j$ appears as or after $\phi_{i+1}$}\end{aligned} \\ 0 \quad & \text{otherwise} \end{cases} \; .
\]
Note that for fixed $\phi$ and $i$, this expression is either $0$ or $1$ (and hence trivially a polynomial).

With the constructs introduced above, we can now integrate $\pdiff$ over an interval $[\phi_i, \phi_{i+1}]$.
It remains to bound the absolute value of the integral for each individual piece.
For this, we introduce a set of $t$ new variables $\xi_1, \ldots, \xi_t$ which will correspond to the absolute value of the integral in the corresponding piece.
\begin{align*}
  \text{\Ak-bounded-interval}_{p, \setP} (\phi, \theta, \xi, i) \ed &\left(\left( -\xi_i \leq \int_{\phi_i}^{\phi_{i + 1}} \pdiff(x) \diff x \right) \wedge \left( \int_{\phi_i}^{\phi_{i + 1}} \pdiff(x) \diff x \leq \xi_i \right)\right) \\
    & \vee (\text{is-active}_{p, \setP} (\phi, i) = 0) \;.
\end{align*}
Note that the above is a valid polynomial constraint because $\pdiff$ depends only on $\theta$ and $x$ for fixed breakpoint order $\phi$ and fixed interval $[\phi_i, \phi_{i+1}]$.
Moreover, recall that by assumption, $\polypdf_{\eps,\theta}(x)$ depends polynomially on both $\theta$ and $x$, and therefore the same holds for $\pdiff$.

We extend the \Ak-check for a single interval to the entire range of \pdiff{} as follows:
\[
  \text{\Ak-bounded-fixed-permutation}_{p, \setP} (\phi, \theta, \xi) \ed \bigwedge_{i = 1}^{t - 1} \text{\Ak-bounded-interval}_{p, \setP} (\phi, \theta, \xi, i) \; .
\]
We now have all the tools to encode the $\Ak$-constraint for a fixed set of intervals:
\begin{align*}
  \text{\Ak-bounded}_{p, \setP} (\theta, \nu, a, b, c, d, \xi) \ed & \left(\sum_{i=1}^{t - 1} \xi_i \leq \nu \right) \; \wedge \; \left(\bigwedge_{i=1}^{t - 1} \left( \xi_i \geq 0 \right) \right) \\
    & \wedge \left( \bigvee_{\phi \in \Phi} \text{ordered}_{p, \setP}(\phi) \wedge \text{\Ak-bounded-fixed-permutation}_{p, \setP} (\phi, \theta, \xi) \right)  \; .
\end{align*}

By construction, the above constraint now satisfies the following:
\begin{lemma}
There exists a vector $\xi \in \R^t$ such that \text{$\Ak$-bounded}$_{p,\setP}(\theta, \nu, a, b, c, d, \xi)$ is true if and only if 
\[
  \sum_{i=1}^K \left| \int_{a_i}^{b_i} \pdiff(x) \diff x \right| \; \leq  \; \nu \; .
\]
Moreover, $\Ak$-bounded$_{p,\setP}$ has less than $6 t^{t+1}$ polynomial constraints.
\end{lemma}
The bound on the number of polynomial constraints follows simply from counting the number of polynomial constraints in the construction described above.

\subsubsection{Complete system of polynomial inequalities}
\label{sec:wbhpolyprogram}

In addition to the \Ak-constraint introduced in the previous subsection, our system of polynomial inequalities contains the following constraints:
\begin{description}
\item[Valid parameters] First, we encode that the mixture parameters we optimize over are valid, i.e., we let 
\[\text{valid-parameters}_S (\theta) \; \ed \; \theta \in S \;. \]
Recall this can be expressed as a Boolean formula over $R$ polynomial predicates of degree at most $D$.
\item[Correct breakpoints] We require that the $d_i$ are indeed the breakpoints of the shape-restricted polynomial $\polypdf_\theta$.
By the assumption, this can be encoded by the following constraint:
\[ 
 \text{correct-breakpoints}_\setP (\theta, d) \; \ed \; \bigwedge_{i = 1}^s (h_i (d_i (\theta), \theta) = 0)\; .
\]
\end{description}

\paragraph{The full system of polynomial inequalities} 
We now combine the constraints introduced above and introduce our entire system of polynomial inequalities:
\begin{align*}
\label{eqn:polyeq}
S_{K, p, \setP, S} (\nu) = \forall a_1, & \ldots a_K,  b_1, \ldots, b_K \,: \\
                                & \exists d_1, \ldots, d_{s}, \xi_1 \ldots \xi_{t} \,: \\
                                & \; \; \text{valid-parameters}_S(\theta) \; \wedge \; \text{correct-breakpoints}_\setP (\theta, d) \; \wedge\; \text{\Ak-bounded}_{p, \setP} (\theta, \nu, a, b, c, d, \xi) \; .
\end{align*}

This system of polynomial inequalities has
\begin{itemize}
\item two levels of quantification, with $2K$ and $s + t$ variables, respectively,
\item $u$ free variables,
\item $R + s + 4 t^{t + 1}$ polynomial constraints,
\item and maximum degree $D$ in the polynomial constraints.
\end{itemize}
Let $\gamma$ be a bound on the free variables, i.e., $\norm{\theta}_2 \leq \gamma$, and let $\lambda$ be a precision parameter.
Then Renegar's algorithm (see Fact \ref{thm:Renegar}) finds a $\lambda$-approximate solution $\theta$ for this system of polynomial inequalities satisfying $\norm{\theta}_2 \leq \gamma$, if one exists, in time
\[
  \left((R + s + 6t^{t + 1}) D\right)^{O(K (s + t) u)} \log \log (3 + \frac{\gamma}{\lambda}) \; .
\]

\subsection{Instantiating the system of polynomial inequalities for GMMs}
\label{sec:polyprogramGMM}
We now show how to use the system of polynomial inequalities developed in the previous subsection for our initial goal: that is, encoding closeness between a well-behaved density estimate and a set of shape-restricted polynomials (see Equation \ref{eq:wellbehavedgoal}).
Our fixed piecewise polynomial ($p$ in the subsection above) will be $\pdens$. 
The set of piecewise polynomials we optimize over (the set $\setP$ in the previous subsection) will be the set $\setP_\epsilon$ of all shape-restricted polynomials $\polypdf_{\epsilon, \theta}$.
Our $S$ (the domain of $\theta$) will be $\paramset_k' \subseteq \paramset_k$, which we define below.
For each $\theta \in S$, we associate it with $\polypdf_{\epsilon, \theta}$.
Moreover:
\begin{itemize}
\item
Define
\[\paramset_{k, \gamma} = \left\{\theta \; \bigg| \, \left(\sum_{i = 1}^k w_i = 1\right) \wedge \big(\forall \, i \in [k] \, : \, (w_i \geq 0) \wedge (\gamma \geq \tau_i > 0) \wedge (-1 \leq \mu_i \leq 1)\big) \right\} \; ,\]
that is, the set of parameters which have bounded means and variances.
$S$ is indeed semi-algebraic, and membership in $S$ can be encoded using $2k + 1$ polynomial predicates, each with degree $D_1 = 1$.
\item
For any fixed parameter $\theta \in \paramset_k$, the shape-restricted polynomial $\polypdf_\theta$ has $s = 2k$ breakpoints by definition, and the breakpoints $d_1 (\theta), \ldots, d_{2k} (\theta)$ of $\polypdf_{\epsilon, \theta}$ occur at 
\[
d_{2i - 1} (\theta) = \frac{1}{\tau_i} \left( \mu_i - 2 \tau_i \log (1 / \epsilon) \right) \; , ~~~ d_{2i} (\theta) = \frac{1}{\tau_i} \left( \mu_i + 2 \tau_i \log (1 / \epsilon) \right) \; , \text{ for all $1 \leq i \leq k$}  \; .
\]
Thus, for all parameters $\theta$, the breakpoints $d_1 (\theta), \ldots, d_{2k} (\theta)$ are the unique numbers so that so that
\[
\tau_i \cdot d_{2i - 1} (\theta) - \left( \mu_i - 2 \tau_i \log (1 / \epsilon) \right) = 0 \; , ~~~ \tau_i \cdot d_{2i} (\theta) - \left( \mu_i + 2 \tau_i \log (1 / \epsilon) \right) = 0 \; , \text{ for all $1 \leq i \leq k$}  \; ,
\]

and thus each of the $d_1 (\theta), \ldots, d_{2k} (\theta)$ can be encoded as a polynomial equality of degree $D_2 = 2$.
\item
Finally, it is straightforward to verify that the map $(x, \theta) \to P_{\epsilon, \theta} (x)$ is a polynomial of degree $D_3 = O(\log 1 / \epsilon)$ in $(x, \theta)$, at any point where $x$ is not at a breakpoint of $P_\theta$.
\end{itemize}
\noindent From the previous subsection, we know that the system of polynomial inequalities $S_{K, \pdens, \setP_\epsilon, \paramset_{k, \gamma}} (\nu)$ has two levels of quantification, each with $O(k)$ variables, it has $k^{O(k)}$ polynomial constraints, and has maximum degree $O(\log 1/ \epsilon)$ in the polynomial constraints. 
Hence, we have shown:
\begin{corollary}
For any fixed $\epsilon$, the system of polynomial inequalities $S_{K, \pdens, \setP_\epsilon, \paramset_{k, \gamma}} (\nu)$ encodes Equation (\ref{eq:wellbehavedgoal}).
Moreover, for all $\gamma, \lambda \geq 0$, Renegar's algorithm $\textsc{Solve-Poly-System} (S_{K, \pdens, \setP_\epsilon, \paramset_{k, \gamma}} (\nu), \lambda, \gamma)$ runs in time $(k \log (1 / \epsilon))^{O(k^4)} \log \log (3 + \frac{\gamma}{\lambda})$.
\end{corollary}

\subsection{Overall learning algorithm}
We now combine our tools developed so far and give an agnostic learning algorithm for the case of well-behaved density estimates (see Algorithm \ref{alg:wbhalg}).
\begin{algorithm}[htb]
\begin{algorithmic}[1]
\Function{Learn-Well-Behaved-GMM}{$k, \epsilon, \delta, \gamma$}
\LineComment{Density estimation. Only this step draws samples.}
\State $\pdens' \gets \textsc{Estimate-Density}(k, \epsilon, \delta)$ \label{line:wbhgetdensity}
\vspace{.3cm}
\LineComment{Rescaling}
\State Let $\pdens$ be a rescaled and shifted version of $\pdens'$ such that the support of $\pdens$ is $[-1, 1]$.
\State Let $\alpha$ and $\beta$ be such that $\pdens(x) = \pdens'\left(\frac{2(x - \alpha)}{\beta - \alpha} - 1 \right)$
\vspace{.3cm}
\LineComment{Fitting shape-restricted polynomials}
\State $K \gets 4k$
\State $\nu \gets \epsilon$
\State $\theta \gets \textsc{Solve-Poly-System}(S_{K, \pdens, \setP_\epsilon, \paramset_{k, \gamma}} (\nu), C \frac{1}{\gamma} \left( \frac{\epsilon}{k} \right)^2, 3k\gamma)$
\While{$\theta$ is ``NO-SOLUTION''}
	\State $\nu \gets 2 \cdot \nu$
  \State $\theta \gets \textsc{Solve-Poly-System}(S_{K, \pdens, \setP_\epsilon, \paramset_{k, \gamma}} (\nu), C \frac{1}{\gamma} \left( \frac{\epsilon}{k} \right)^2, 3k\gamma)$
\EndWhile
\vspace{.3cm}
\LineComment{Fix the parameters}
\For{$i = 1, \ldots, k$}
	\State if $\tau_i \leq 0$, set $w_i \gets 0$ and set $\tau_i$ to be arbitrary but positive.
\EndFor
\State Let $W = \sum_{i = 1}^k w_i$
\For{$ i = 1, \ldots, k$}
	\State $w_i \gets w_i / W$
\EndFor
\vspace{.3cm}
\LineComment{Undo the scaling}
\State $w_i' \gets w_i$ \label{line:wbhscaleback1}
\State $\mu_i' \gets \frac{(\mu_i + 1)(\beta - \alpha)}{2} + \alpha$
\State $\tau_i' \gets \frac{\tau_i}{\beta - \alpha}$ \label{line:wbhscaleback3}
\State \textbf{return} $\theta'$
\EndFunction
\end{algorithmic}
\caption{Algorithm for learning a mixture of Gaussians in the well-behaved case.}
\label{alg:wbhalg}
\end{algorithm}

\subsection{Analysis}
Before we prove correctness of \textsc{Learn-Well-Behaved-GMM}, we introduce two auxiliary lemmas.

An important consequence of the well-behavedness assumption (see Definition \ref{def:wellbehaved}) are the following robustness properties.
\begin{lemma}[Parameter stability]
\label{lem:robust}
Fix $2 \geq \epsilon > 0$. Let the parameters $\theta, \theta' \in \paramset_k$ be such that (i) $\tau_i, \tau'_i  \leq \gamma$ for all $i \in [k]$ and (ii) $\| \theta - \theta' \|_2 \leq C \frac{1}{\gamma} \left( \frac{\epsilon}{k} \right)^2$, for some universal constant $C$. Then 
\[
  \| \mixpdf_{\theta} - \mixpdf_{\theta'}\|_1 \; \leq \; \epsilon \; .
\]
\end{lemma}

Before we prove this lemma, we first need a calculation which quantifies the robustness of the standard normal pdf to small perturbations.
\begin{lemma}
\label{lem:perturbstandard}
For all $2 \geq \epsilon > 0$, there is a $\delta_1 = \delta_1 (\epsilon) = \frac{\epsilon}{20 \sqrt{\log (1 / \epsilon)}} \geq O(\epsilon^2)$ so that for all $\delta \leq \delta_1$, we have $\| \normalpdf (x) - \normalpdf (x + \delta) \|_1 \leq O(\epsilon)$.
\end{lemma}
\begin{proof}
Note that if $\epsilon > 2$ this claim holds trivially for all choices of $\delta$ since the $L_1$-distance between two distributions can only ever be $2$.
Thus assume that $\epsilon \leq 2$.
Let $I$ be an interval centered at $0$ so that both $\normalpdf (x)$ and $\normalpdf (x + \delta)$ assign $1 - \frac{\eps}{2}$ weight on this interval.
By standard properties of Gaussians, we know that $|I| \leq 10 \sqrt{\log (1 / \epsilon)}$.
We thus have
\[\| \normalpdf (x) - \normalpdf (x + \delta) \|_1 \leq \int_I \left| \normalpdf (x) - \normalpdf(x + \delta) \right| dx + \epsilon \; .\]
By Taylor's theorem we have that for all $x$,
\[\left| e^{-(x + \delta)^2 /2} - e^{-x^2 / 2} \right| \leq C \cdot \delta\]
for some universal constant $C = \max_{x \in \R} \frac{\diff}{\diff x} (e^{-\frac{x^2}{2}}) \leq 1$.
Since we choose $\delta_1 \leq \frac{\epsilon}{20 \sqrt{\log (1 / \epsilon)}}$, we must have that
\[\| \normalpdf (x) - \normalpdf (x + \delta) \|_1 \leq O(\epsilon) \; , \]
as claimed.
\end{proof}

\begin{proof}[Proof of Lemma \ref{lem:robust}]
Notice the $\ell_2$ guarantee of Renegar's algorithm (see Fact \ref{thm:Renegar}) also trivially implies an $\ell_\infty$ guarantee on the error in the parameters $\theta$; that is, for all $i$, we will have that the weights, means, and variances of the two components differ by at most $C \frac{1}{\gamma} \left( \frac{\epsilon}{k} \right)^2$.
By repeated applications of the triangle inequality to the quantity in the lemma, it suffices to show the three following claims:
\begin{itemize}
\item
For any $\mu, \tau$,
\[\| w_1 \normalpdf_{\mu, \tau}(x) - w_2 \normalpdf_{\mu, \tau}(x) \|_1 \leq \frac{\epsilon}{3k}\]
if $|w_1 - w_2| \leq C \frac{1}{\gamma} \left( \frac{\epsilon}{k} \right)^2$.
\item
For any $\tau \leq \gamma$,
\[\|\normalpdf_{\mu_1, \tau}(x) -\normalpdf_{\mu_2, \tau}(x) \|_1 \leq \frac{\epsilon}{3k}\]
if $|\mu_1 - \mu_2| \leq  C \frac{1}{\gamma} \left( \frac{\epsilon}{k} \right)^2$.
\item
For any $\mu$,
\[\|\normalpdf_{\mu, \tau_1}(x) - \normalpdf_{\mu, \tau_2}(x) \|_1 \leq \frac{\epsilon}{3k}\]
if $|\tau_1 - \tau_2| \leq C \frac{1}{\gamma} \left( \frac{\epsilon}{k} \right)^2$.
\end{itemize}
The first inequality is trivial, for $C$ sufficiently small.
The second and third inequalities follow from a change of variables and an application of Lemma  \ref{lem:perturbstandard}.
\end{proof}

Recall that our system of polynomial inequalities only considers mean parameters in $[-1, 1]$.
The following lemma shows that this restriction still allows us to find a good approximation once the density estimate is rescaled to $[-1, 1]$.
\begin{lemma}[Restricted means]
\label{lem:restrictedmeans}
Let $g: \R \to \R$ be a function supported on $[-1, 1]$, i.e., $g(x) = 0$ for $x \notin [-1, 1]$.
Moreover, let $\theta^* \in \paramset_k$.
Then there is a $\theta' \in \paramset_k$ such that $\mu_i' \in [-1, 1]$ for all $i \in [k]$ and
\[
  \norm{g - \mixpdf_{\theta'}}_1 \; \leq \; 5 \cdot \norm{g - \mixpdf_{\theta^*}}_1 \; .
\]
\end{lemma}
\begin{proof}
Let $A = \{ i \, | \, \mu^*_i \in [-1, 1] \}$ and $B = [k] \setminus A$.
Let $\theta'$ be defined as follows:
\begin{itemize}
\item $w'_i = w^*_i$ for all $i \in [k]$.
\item $\mu'_i = \mu^*_i$ for $i \in A$ and $\mu_i' = 0$ for $i \in B$.
\item $\tau'_i = \tau^*_i$ for all $i \in [k]$.
\end{itemize}
From the triangle inequality, we have
\begin{equation}
\label{eq:restrictedmean}
  \norm{g - \mixpdf_{\theta'}}_1 \; \leq \; \norm{g - \mixpdf_{\theta^*}}_1 + \norm{\mixpdf_{\theta^*} - \mixpdf_{\theta'}}_1 \; .
\end{equation}
Hence it suffices to bound $\norm{\mixpdf_{\theta^*} - \mixpdf_{\theta'}}_1$.

Note that for $i \in B$, the corresponding $i$-th component has at least half of its probability mass outside $[-1, 1]$.
Since $g$ is zero outside $[-1, 1]$, this mass of the $i$-th component must therefore contribute to the error $\norm{g - \mixpdf_{\theta^*}}_1$.
Let $\mathbbm{1}[x \notin [-1, 1]]$ be the indicator function of the set $\R \setminus [-1, 1]$.
Then we get
\[
  \norm{g - \mixpdf_{\theta^*}}_1 \; \geq \; \norm{\mixpdf_{\theta^*} \cdot \mathbbm{1}[x \notin [-1, 1]]}_1 \; \geq \; \frac{1}{2} \norm*{\sum_{i \in B} w^*_i \cdot \normalpdf_{\mu^*_i, \tau^*_i}}_1 \; .
\]
For $i \in A$, the mixture components of $\mixpdf_{\theta^*}$ and $\mixpdf_{\theta'}$ match.
Hence we have
\begin{align*}
  \norm{\mixpdf_{\theta^*} - \mixpdf_{\theta'}}_1 \; &= \; \norm*{\sum_{i \in B} w^*_i \cdot \normalpdf_{\mu^*_i, \tau^*_i} \; - \; \sum_{i \in B} w'_i \cdot \normalpdf_{\mu'_i, \tau'_i}}_1 \\
      &\leq \; \norm*{\sum_{i \in B} w^*_i \cdot \normalpdf_{\mu^*_i, \tau^*_i}}_1 \; + \; \norm*{\sum_{i \in B} w'_i \cdot \normalpdf_{\mu'_i, \tau'_i}}_1 \\
      &= \; 2 \cdot \norm*{\sum_{i \in B} w^*_i \cdot \normalpdf_{\mu^*_i, \tau^*_i}}_1 \\
      &\leq \; 4 \cdot \norm{g - \mixpdf_{\theta^*}}_1 \; .
\end{align*}
Combining this inequality with \eqref{eq:restrictedmean} gives the desired result.
\end{proof}

We now prove our main theorem for the well-behaved case.
\begin{theorem}
Let $\delta, \epsilon, \gamma > 0$, $k \geq 1$, and let $f$ be the pdf of the unknown distribution.
Moreover, assume that the density estimate $\pdens'$ obtained in Line \ref{line:wbhgetdensity} of Algorithm \ref{alg:wbhalg} is $\gamma$-well-behaved.
Then the algorithm $\textsc{Learn-Well-Behaved-GMM}(k, \epsilon, \delta, \gamma)$ returns a set of GMM parameters $\theta'$ such that
\[
  \norm{\mixpdf_{\theta'} - f}_1 \; \leq \; 60 \cdot \OPT_k + \eps 
\]
with probability $1 - \delta$.
Moreover, the algorithm runs in time
\[
  \left(k \cdot \log\frac{1}{\eps}\right)^{O(k^4)} \!\! \cdot \, \log\frac{1}{\eps} \, \cdot \, \log\log\frac{k\gamma}{\eps} \; + \; \Otilde\left(\frac{k}{\eps^2}\right) \; .
\]
\end{theorem}
\begin{proof}First, we prove the claimed running time.
From Fact \ref{thm:ADLS}, we know that the density estimation step has a time complexity of $\Otilde(\frac{k}{\eps^2})$.
Next, consider the second stage where we fit shape-restricted polynomials to the density estimate.
Note that for $\nu = 3$, the system of polynomial inequalities $S_{\pdens, \setP_\epsilon} (\nu)$ is trivially satisfiable because the $\Ak$-norm is bounded by the $L_1$-norm and the $L_1$-norm between the two (approximate) densities is at most $2 + O(\eps)$.
Hence the while-loop in the algorithm takes at most $O(\log \frac{1}{\eps})$ iterations.
Combining this bound with the size of the system of polynomial inequalities (see Subsection \ref{sec:wbhpolyprogram}) and the time complexity of Renegar's algorithm (see Fact \ref{thm:Renegar}), we get the following running time for solving all systems of polynomial inequalities proposed by our algorithm:
\[
  \left(k \cdot \log\frac{1}{\eps}\right)^{O(k^4)} \!\! \cdot \, \log\log\frac{k\gamma}{\eps} \, \cdot \, \log\frac{1}{\eps} \; .
\]
This proves the stated running time.

Next, we consider the correctness guarantee.
We condition on the event that the density estimation stage succeeds, which occurs with probability $1 - \delta$ (Fact \ref{thm:ADLS}).
Then we have
\[
  \norm{f - \pdens'}_1 \leq 4\cdot \OPT_k + \eps \; .
\]
Moreover, we can assume that the rescaled density estimate $\pdens$ is $\gamma$-well-behaved.
Recalling Definition \ref{def:wellbehaved}, this means that there is a set of GMM parameters $\theta \in \paramset_k$ such that $\tau_i \leq \gamma$ for all $i \in [k]$ and
\begin{align*}
  \norm{\pdens - \mixpdf_{\theta}}_1 \; &= \; \min_{\theta^* \in \paramset_k} \norm{\pdens - \mixpdf_{\theta^*}}_1 \\
      &= \; \min_{\theta^* \in \paramset_k} \norm{\pdens' - \mixpdf_{\theta^*}}_1  \\
      &\leq \; \min_{\theta^* \in \paramset_k} \norm{\pdens' - f}_1 \, + \, \norm{f - \mixpdf_{\theta^*}}_1  \\
      &\leq \; 4 \cdot \OPT_k + \eps \; + \; \min_{\theta^* \in \paramset_k} \norm{f - \mixpdf_{\theta^*}}_1 \\
      &\leq \; 5 \cdot \OPT_k + \eps \; .
\end{align*}
Applying the triangle inequality again, this implies that
\[
  \norm{\pdens - \polypdf_{\eps, \theta}}_1 \; \leq \; \norm{\pdens - \mixpdf_\theta}_1 \, + \, \norm{\mixpdf_\theta - \polypdf_{\eps, \theta}}_1 \; \leq \; 5 \cdot \OPT_k + 2 \eps \; .
\]
This almost implies that $S_{\pdens, \setP_\epsilon} (\nu)$ is feasible for $\nu \geq 5 \cdot \OPT_k + 2 \eps$.
However, there are two remaining steps.
First, recall that the system of polynomial inequalities restricts the means to lie in $[-1, 1]$.
Hence we use Lemma \ref{lem:restrictedmeans}, which implies that there is a $\theta^\dagger \in \paramset_k$ such that $\mu^\dagger_i \in [-1, 1]$ and
\[
  \norm{\pdens - \polypdf_{\eps, \theta^\dagger}}_1 \; \leq \; 25 \cdot \OPT_k + 10 \eps \; .
\]
Moreover, the system of polynomial inequalities works with the $\Ak$-norm instead of the $L_1$-norm.
Using Lemma \ref{lem:Akapprox}, we get that
\[
  \norm{\pdens - \polypdf_{\eps, \theta^\dagger}}_\Ak \leq \; \norm{\pdens - \polypdf_{\eps, \theta^\dagger}}_1 \; .
\]
Therefore, in some iteration when
\[
  \nu \; \leq \; 2\cdot (25 \cdot \OPT_k + 10 \eps) \; = \; 50 \cdot \OPT_k + 20 \eps
\]
the system of polynomial inequalities $S_{\pdens, \setP_\epsilon, \paramset_{k, \gamma}} (\nu)$ become feasible and Renegar's algorithm guarantees that we find parameters $\theta'$ such that $\norm{\theta' - \theta^\dagger}_2 \leq \frac{\eps}{\gamma}$ for some $\theta^\dagger \in \paramset_k$ and
\[
  \norm{\pdens - \mixpdf_{\theta^\dagger}}_\Ak \; \leq \; 50 \cdot \OPT_k + O(\eps) \; .
\]
Note that we used well-behavedness here to ensure that the precisions in $\theta^\dagger$ are bounded by $\gamma$.
Let $\theta$ be the parameters we return.
It is not difficult to see that $\| \theta - \theta^\dagger \|_2 \leq  \frac{2\eps}{\gamma}$.
We convert this back to an $L_1$ guarantee via Lemma \ref{lem:Akapprox}:
\[
  \norm{\pdens - \mixpdf_{\theta^\dagger}}_1 \; \leq \; 56 \cdot \OPT_k + O(\eps) \; .
\]
Next, we use parameter stability (Lemma \ref{lem:robust}) and get
\[
  \norm{\pdens - \mixpdf_{\theta}}_1 \; \leq \; 56 \cdot \OPT_k + O(\eps) \; .
\]
We now relate this back to the unknown density $f$.
Let $\theta'$ be the parameters $\theta$ scaled back to the original density estimate (see Lines \ref{line:wbhscaleback1} to \ref{line:wbhscaleback3} in Algorithm \ref{alg:wbhalg}).
Then we have
\[
  \norm{\pdens' - \mixpdf_{\theta'}}_1 \; \leq \; 56 \cdot \OPT_k + O(\eps) \; .
\]
Using the fact that $\pdens'$ is a good density estimate, we get
\begin{align*}
  \norm{f - \mixpdf_{\theta'}}_1 \; &\leq \; \norm{f - \pdens'}_1 \, + \, \norm{\pdens' - \mixpdf_{\theta'}}_1 \\
    &\leq \; 4 \cdot \OPT_k + \eps \, + \, 56 \cdot \OPT_k + O(\eps) \\
    &\leq \; 60 \cdot \OPT_k + O(\eps) \; .
\end{align*}
As a final step, we choose an internal $\eps'$ in our algorithm so that the $O(\eps')$ in the above guarantee becomes bounded by $\eps$.
This proves the desired approximation guarantee.
\end{proof}

%% file: clustering.tex

\section{General algorithm}
\label{sec:generalalg}
\subsection{Preliminaries}
As before, we let $\pdens$ be the piecewise polynomial returned by \textsc{Learn-Piecewise-Polynomial} (see Fact \ref{thm:ADLS}).
Let $I_0, \ldots, I_{s + 1}$ be the intervals defined by the breakpoints of $\pdens$.
Recall that $\pdens$ has degree $O(\log 1/\epsilon)$ and has $s + 2 = O(k)$ pieces.
Furthermore, $I_0$ and $I_{s + 1}$ are unbounded in length, and on these intervals $\pdens$ is zero.
By rescaling and translating, we may assume WLOG that $\cup_{i = 1}^s I_i$ is $[-1, 1]$.

Recall that $\setI$ is defined by the set of intervals $\{ I_1, \ldots, I_s \}$. 
We know that $s = O(k)$.
Intuitively, these intervals capture the different scales at which we need to operate.
We formalize this intuition below.

\begin{definition}
For any Gaussian $\normalpdf_{\mu, \tau}$, let $L (\normalpdf_{\mu, \tau})$ be the interval centered at $\mu$ on which $\normalpdf_{\mu, \tau}$ places exactly $W$ of its weight, where $0 < W < 1$ is a universal constant we will determine later.
By properties of Gaussians, there is some absolute constant $\omega > 0$ such that $\normalpdf_{\mu, \tau} (x) \geq \omega \tau$ for all $x \in L (\normalpdf_{\mu, \tau})$.
\end{definition}

\begin{definition}
\label{def:admissible}
Say a Gaussian $\normalpdf_{\mu, \tau}$ is \emph{admissible} if (i) $\normalpdf_{\mu, \tau}$ places at least $1/2$ of its mass in $[-1, 1]$, and (ii) there is a $J \in \setI$ so that $|J \cap L(\normalpdf_{\mu, \tau}) | \geq 1 / (8 s \tau)$ and so that 
\[\tau \leq \frac{1}{|J|} \cdot \phi, \]
where
\[
\phi \; = \; \phi(\epsilon, k) \; \ed \; \frac{32 k }{\omega \epsilon} m (m + 1)^2 \cdot (\sqrt{2} + 1)^m \; ,
\]
where $m$ is the degree of $\pdens$.
We call the interval $J \in \setI$ satisfying this property on which $\normalpdf_{\mu, \tau}$ places most of its mass its \emph{associated interval}.

Fix $\theta \in \paramset_k$. We say the $\ell$-th component is \emph{admissible} if the underlying Gaussian is admissible and moreover $w_\ell \geq \epsilon / k$.
\end{definition}

Notice that since $m = O(\log (1 / \epsilon))$, we have that $\phi (\epsilon, k) = \poly (1 / \epsilon, k)$.

\begin{lemma}[No Interaction Lemma]
\label{lem:cluster-struct}
Fix $\theta \in \paramset_k$.
Let $S_\good (\theta) \subseteq [k]$ be the set of $\ell \in [k]$ whose corresponding mixture component is admissible, and let $S_\bad (\theta)$ be the rest.
Then, we have
\[\left\| \mixpdf_\theta  - \pdens \right\|_1 \geq \left\| \sum_{\ell \in S_\good (\theta)} w_\ell \cdot \normalpdf_{\mu_\ell, \tau_\ell} - \pdens \right\|_1 +  \frac{1}{2} \sum_{\ell \in S_\bad (\theta)} w_\ell \; - 2 \epsilon.\]
\end{lemma}
We briefly remark that the constant $\frac{1}{2}$ we obtain here is somewhat arbitrary; by choosing different universal constants above, one can obtain any fraction arbitrarily close to one, at a minimal loss.
\begin{proof}
Fix $\ell \in S_\bad (\theta)$, and denote the corresponding component $\normalpdf_\ell$. 
Recall that it has mean $\mu_\ell$ and precision $\tau_\ell$. Let $L_\ell = L (\normalpdf_\ell)$.

Let $\mixpdf^{- \ell}_\theta (x) = \sum_{i \neq \ell} w_i \normalpdf_{\mu_i, \tau_i} (x)$ be the density of the mixture without the $\ell$-th component.
We will show that
\[
\left\| \mixpdf_\theta  - \pdens \right\|_1 \geq \| M_\theta^{-\ell} - \pdens \|_1 + \frac{1}{2} w_\ell - \frac{2 \epsilon}{k}.
\]
It suffices to prove this inequality because then we may repeat the argument with a different $\ell' \in S_\bad (\theta)$ until we have subtracted out all such $\ell$, and this yields the claim in the lemma.

If $w_\ell \leq \epsilon / k$ then this statement is obvious.
If $\normalpdf_\ell$ places less than half its weight on $[-1, 1]$, then this is also obvious.
Thus we will assume that $w_\ell > \epsilon / k$ and $\normalpdf_\ell$ places at least half its weight on $[-1, 1]$.

Let $\setI_\ell$ be the set of intervals in $\setI$ which intersect $L_\ell$.
We partition the intervals in $\setI_\ell$ into two groups:
\begin{enumerate}
\item
Let $\mathcal{L}_1$ be the set of intervals $J \in \setI_\ell$ so that $|J \cap L_\ell | \leq 1 / (8 s \tau_\ell)$.
\item
Let $\mathcal{L}_2$ be the set of intervals $J \in \setI_\ell$ not in $\mathcal{L}_1$ so that 
\[\tau_\ell >  \frac{1}{|J|} \cdot \phi \;. \]
\end{enumerate}
By the definition of admissibility, this is indeed a partition of $\setI_\ell$.

We have
\begin{align*}
\left\| \mixpdf_\theta - \pdens \right\|_1 &= \left\| \mixpdf^{-\ell}_\theta + \normalpdf_\ell  - \pdens \right\|_1 \\
&= \int_{L_\ell} \left| \mixpdf^{-\ell}_\theta (x) + w_\ell \normalpdf_\ell(x)  - \pdens (x) \right| \diff x + \int_{L_\ell^c} \left| \mixpdf^{-\ell}_\theta (x) + w_\ell \normalpdf_\ell(x)  - \pdens (x) \right| \diff x \\
&\geq \int_{L_\ell} \left| \mixpdf^{-\ell}_\theta (x) + w_\ell \normalpdf_\ell(x)  - \pdens (x) \right| \diff x + \int_{L_\ell^c} \left| \mixpdf^{-\ell}_\theta (x) - \pdens (x) \right| \diff x - w_\ell \int_{L_\ell^c} \normalpdf_\ell(x) \diff x \\
&\geq  \int_{L_\ell} \left| \mixpdf^{-\ell}_\theta (x) + w_\ell \normalpdf_\ell(x)  - \pdens (x) \right| \diff x + \int_{L_\ell^c} \left| \mixpdf^{-\ell}_\theta (x) - \pdens (x) \right| \diff x - (1 - W) w_\ell \; .
\end{align*}
We split the first term on the RHS into two parts, given by our partition:
\begin{align*}
 \int_{L_\ell} \left| \mixpdf^{-\ell}_\theta (x) + w_\ell \normalpdf_\ell(x)  - \pdens (x) \right| \diff x = \sum_{J \in \mathcal{L}_1} \int_{J \cap L_\ell} \left| \mixpdf^{-\ell}_\theta (x) + w_\ell \normalpdf_\ell(x)  - \pdens (x) \right| \diff x \\
+ \sum_{J \in \mathcal{L}_2} \int_{J \cap L_\ell} \left| \mixpdf^{-\ell}_\theta (x) + w_\ell \normalpdf_\ell(x)  - \pdens (x) \right| \diff x \; .
\end{align*}
We lower bound the contribution of each term separately.

\paragraph{(1)} We first bound the first term.
Since for each $J \in \mathcal{L}_1$ we have $|J \cap L_\ell| \leq 1 / (8 s \tau_\ell)$, we know that 
\begin{equation}
\label{eq:case1-Nbound}
\int_{J \cap L_\ell} \normalpdf_\ell(x) \diff x \leq \frac{1}{8 s}
\end{equation}
and so
\begin{align*}
\sum_{J \in \mathcal{L}_1} \int_{J \cap L_\ell} \left| \mixpdf^{-\ell}_\theta (x) + w_\ell \normalpdf_\ell(x)  - \pdens (x) \right| \diff x &\geq \sum_{J \in \mathcal{L}_1} \int_{J \cap L_\ell} \left| \mixpdf^{-\ell}_\theta (x)  - \pdens (x) \right| \diff x - |\mathcal{L}_1| \cdot w_\ell  \cdot \frac{1}{8 s} \\
& \geq \sum_{J \in \mathcal{L}_1} \int_{J \cap L_\ell} \left| \mixpdf^{-\ell}_\theta (x)  - \pdens (x) \right| \diff x - \frac{1}{8} w_\ell 
\end{align*}
since $\setI$ and thus $\mathcal{L}_1$ contains at most $s$ intervals.

\paragraph{(2)} We now consider the second term.
Fix a $J \in \mathcal{L}_2$, and let $p_J$ be the polynomial which is equal to $\pdens$ on $J$.
Since $\int \pdens \leq 1 + \eps \leq 2$ (as otherwise its $L_1$-distance to the unknown density would be more than $\epsilon$) and $\pdens$ is nonnegative, we also know that $\int_J p_J \leq 2$.
We require the following fact (see \cite{ADLS15}):
\begin{fact}
\label{lem:polybound}
Let $p(x) = \sum_{j = 0}^m c_j x^j$ be a degree-$m$ polynomial so that $p \geq 0$ on $[-1, 1]$ and $\int_{-1}^1 p \leq \beta$. Then $\max_i \abs{c_i} \leq \beta \cdot (m + 1)^2 \cdot (\sqrt{2} + 1)^m$.
\end{fact}
Consider the shifted polynomial $q_J (u) = p_J (u \cdot (b_J - a_J) / 2 + (b_J + a_J) / 2)$ where $J = [a_J, b_J]$.
By applying Fact \ref{lem:polybound} to $q_J$ and noting that $\int_{-1}^1 q_J = (2 / |J|) \cdot \int_J p_J$, we conclude that the coefficients of $q_J$ are bounded by
\[\frac{4}{|J|} \cdot (m + 1)^2 \cdot (\sqrt{2} + 1)^m\]
and thus
\[| q_J (u)| \leq \frac{4}{|J|} \cdot m(m + 1)^2 \cdot (\sqrt{2} + 1)^m \]
for all $u \in [-1, 1]$, and so therefore the same bound applies for $p_J (x)$ for all $x \in J$.

But notice that since we assume that $J \in \mathcal{L}_2$, it follows that for all $x \in J \cap L_\ell$, we have that 
\[
\normalpdf_\ell (x) \geq 8 \frac{k}{\epsilon} p_J(x) \;, \]
and so in particular $w_\ell \normalpdf_\ell (x) \geq 8 p_J(x)$ for all $x \in J \cap L_\ell$.
Hence we have
\begin{align*}
\int_{J \cap L_\ell} \left| \mixpdf^{-\ell}_\theta (x) + w_\ell \normalpdf_\ell(x)  - \pdens (x) \right| \diff x &= \int_{J \cap L_\ell} \mixpdf^{-\ell}_\theta (x) + w_\ell \normalpdf_\ell(x)  - p_J (x) \diff x \\
&\geq \int_{J \cap L_\ell} \left| \mixpdf^{-\ell}_\theta (x)   - p_J (x) \right| \diff x + \int_{J \cap L_\ell} \frac{7}{8} w_\ell \normalpdf_\ell(x)  - p_J (x) \diff x \\
&\geq  \int_{J \cap L_\ell} \left| \mixpdf^{-\ell}_\theta (x)   - p_J (x) \right| \diff x + \frac{3 w_\ell}{4} \int_{J \cap L_\ell} \normalpdf_\ell (x) \diff x \; . \\
\end{align*}
where the second line follows since $\mixpdf^{-\ell}_\theta (x) + w_\ell \normalpdf_\ell  - p_J(x) \geq \left| \mixpdf^{-\ell}_\theta (x)   - p_J (x)\right| + \frac{7}{8} w_\ell \normalpdf_\ell(x)  - p_J (x)$ for all $x \in J \cap L_\ell$.

Thus
\begin{equation}
\label{eq:case3}
\begin{split}
\sum_{J \in \mathcal{L}_2} \int_{J \cap L_\ell} \big| \mixpdf^{-\ell}_\theta (x) & + w_\ell \normalpdf_\ell(x)  - \pdens(x) \big| \diff x \;\, \geq \\
&\sum_{J \in \mathcal{L}_2} \left( \int_{J \cap L_\ell} \left| \mixpdf^{-\ell}_\theta (x)   - \pdens (x) \right| \diff x + \frac{3 w_\ell}{4} \int_{J \cap L_\ell} \normalpdf_\ell (x) \diff x \right) \; .
\end{split}
\end{equation}
Moreover, by Equation (\ref{eq:case1-Nbound}), we know that 
\begin{align*}
\sum_{J \in \mathcal{L}_2} \int_{J \cap L_\ell} \normalpdf_\ell (x) \diff x &= \int_{L_\ell} \normalpdf_\ell (x) dx - \sum_{J \in \mathcal{L}_1 } \int_{J \cap L_\ell} \normalpdf_\ell (x) \diff x \\
& \geq W - \frac{1}{8} \;,
\end{align*}
since $\mathcal{L}_1$ contains at most $s$ intervals. 
Thus, the RHS of Equation (\ref{eq:case3}) must be lower bounded by
\[
\sum_{J \in \mathcal{L}_2} \int_{J \cap L_\ell} \left| \mixpdf^{-\ell}_\theta (x)   - \pdens (x) \right| \diff x + \frac{3}{4} \left( W - \frac{1}{8}\right) w_\ell \;.
\]

\paragraph{Putting it all together.}
Hence, we have
\begin{align*}
\int_{L_\ell} | M_\theta (x) - \pdens (x) | \diff x &=   \sum_{J \in \mathcal{L}_1} \int_{J \cap L_\ell} | M_\theta (x) - \pdens (x) | \diff x + \sum_{J \in \mathcal{L}_2} \int_{J \cap L_\ell} | M_\theta (x) - \pdens (x) | \diff x \\
& \geq \sum_{J \in \mathcal{L}_1} \int_{J \cap L_\ell} | M_\theta^{-\ell} (x) - \pdens (x) | \diff x + \sum_{J \in \mathcal{L}_2} \int_{J \cap L_\ell} | M_\theta^{-\ell} (x) - \pdens (x) | \diff x  \\
& \qquad\qquad + \left[ \frac{3}{4} \left( W - \frac{1}{8}\right) - \frac{1}{8} \right]  w_\ell \\
&\geq \int_{L_\ell} | M_\theta^{-\ell} (x) - \pdens (x) | \diff x + \left[ \frac{3}{4} \left( W - \frac{1}{8}\right) - \frac{1}{8} \right]  w_\ell \; .
\end{align*}

We therefore have
\begin{align*}
\| M_\theta - \pdens \|_1 &= \int_{L_\ell} | M_\theta (x) - \pdens (x) | \diff x + \int_{L_\ell^c} |M_\theta (x) - \pdens (x) | \diff x \\
&\geq \int_{L_\ell} | M_\theta (x) - \pdens (x) | \diff x + \int_{L_\ell^c} \left| M_\theta^{-\ell} (x) - \pdens (x) \right| \diff x - \int_{L_\ell^c} w_\ell \normalpdf_\ell(x) \diff x\\
&\geq \int_{L_\ell} | M_\theta (x) - \pdens (x) | \diff x + \int_{L_\ell^c} \left| M_\theta^{-\ell} (x) - \pdens (x) \right| \diff x - (1 - W) w_\ell\\
&\geq \int_{L_\ell} | M_\theta^{-\ell} (x) - \pdens (x) | \diff x + \left( \frac{7}{4} W - \frac{39}{32} \right) w_\ell + \int_{L_\ell^c} \left| M_\theta^{-\ell} (x) - \pdens (x) \right| \diff x \\
& = \| M_\theta^{-\ell} - \pdens \|_1 + \frac{1}{2} w_\ell \;,
\end{align*}
when we set $W = 55/56$.
\end{proof}

\subsection{A parametrization scheme for a single Gaussian} 
Intuitively, Lemma \ref{lem:cluster-struct} says that for any $\theta \in \paramset_k$, there are some components which have bounded variance and which can be close to $\pdens$ (the components in $S_1$), and the remaining components, which may have unbounded variance but which will be far away from $\pdens$.
Since we are searching for a $k$-GMM which is close to $\pdens$, in some sense we should not have to concern ourselves with the latter components since they cannot meaningfully interact with $\pdens$.
Thus we only need find a suitably robust parametrization for admissible Gaussians.

Such a parametrization can be obtained by linearly transforming the domain so that the associated interval gets mapped to $[-1, 1]$.
Formally, fix a Gaussian $\normalpdf_{\mu,\tau}$ and an interval $J$.
Then it can be written as
\begin{equation}
\label{eq:rescaled}
\normalpdf_{\mu, \tau} (x) = \frac{\tilde{\tau}}{|J| / 2} \normalpdf \left( \tilde{\tau} \cdot \frac{x - \midp (J)}{|J| / 2} - \tilde{\mu} \right) \; ,
\end{equation}
for some unique $\tilde{\mu}$ and $\tilde{\tau}$, where for any interval $I$, we define $\midp(I)$ to denote its midpoint. Call these the \emph{rescaled mean with respect to $J$} and \emph{rescaled precision with respect to $J$} of $\normalpdf$, respectively.
Concretely, given $\mu$, $\tau$, and an interval $J$, the rescaled variance and mean with respect to $J$ are defined to be
\[\tilde{\tau} = \frac{|J|}{2} \tau \;, ~~~~~~~~~~~ \tilde{\mu} = \frac{\tilde{\tau}}{|J| / 2} \left( \mu - \midp (J) \right) \; . \]
For any $\tilde{\mu}, \tilde{\tau}$, we let $\normalpdf^{r, J}_{\tilde{\mu}, \tilde{\tau}}(x)$ denote the function given by the RHS of Equation (\ref{eq:rescaled}).
The following two lemmas says that these rescaled parameters have the desired robustness properties.

\begin{lemma}
\label{lem:rescaled-robust}
Let $\normalpdf^{r, J}_{\tilde{\mu}, \tilde{\tau}}$ be an admissible Gaussian with rescaled mean $\tilde{\mu}$ and rescaled precision $\tilde{\tau}$ with respect to its associated interval $J \in \setI$.
Then $\tilde{\mu} \in [-\frac{2 s \phi}{\omega} , \frac{2 s \phi}{\omega}]$ and $\sqrt{2 \pi} \cdot \omega / (16 s) \leq \tilde{\tau} \leq \phi / 2$.
\end{lemma}
\begin{proof}
We first show that $\sqrt{2\pi} \cdot \omega / (16s) \leq \tilde{\tau} \leq \phi / 2$. 
That the rescaled variance is bounded from above follows from a simple change of variables and the definition of admissibility.
By the definition of admissibility, we also know that 
\begin{align*} 
\int_J \normalpdf^{r, J}_{\tilde{\mu}, \tilde{\tau}} \diff x &\geq \int_{J \cap L(\normalpdf^{r, J}_{\tilde{\mu}, \tilde{\tau}})} \normalpdf^{r, J}_{\tilde{\mu}, \tilde{\tau}} \diff x \\
&\geq \omega \tau \cdot |J \cap L(\normalpdf^{r, J}_{\tilde{\mu}, \tilde{\tau}})| \\
&\geq \frac{\omega}{8 s} \; .
\end{align*}
Furthermore, we trivially have
\[
|J| \cdot \frac{\tau}{\sqrt{2\pi}} \; \geq \; \int_J \normalpdf^{r, J}_{\tilde{\mu}, \tilde{\tau}} \diff x  \; .
\]
Thus, the precision $\tau$ must be at least $\sqrt{2 \pi} \omega / (8 s |J|)$, and so its rescaled precision must be at least $\sqrt{2 \pi} \omega / (16 s)$, as claimed.

We now show that $\tilde{\mu} \in [-\frac{2 s \phi}{\omega} , \frac{2 s \phi}{\omega}]$.
Because $\normalpdf^{r, J}_{\tilde{\mu}, \tilde{\tau}}$ is an admissible Gaussian with associated interval $J$, we know that $| J \cap L(\normalpdf^{r, J}_{\tilde{\mu}, \tilde{\tau}}) | \geq 1 / (8 s \tau)$.
Moreover, we know that on $J \cap L(\normalpdf^{r, J}_{\tilde{\mu}, \tilde{\tau}})$, we have $\normalpdf^{r, J}_{\tilde{\mu}, \tilde{\tau}} (x) \geq \omega \tau$.
Thus in particular
\[ \int_J \normalpdf^{r, J}_{\tilde{\mu}, \tilde{\tau}} \diff x \geq \int_{J \cap L(\normalpdf^{r, J}_{\tilde{\mu}, \tilde{\tau}})} \normalpdf^{r, J}_{\tilde{\mu}, \tilde{\tau}} \diff x \geq \frac{\omega}{8 s} \; . \]

Define $\tilde{J}$ to be the interval which is of length $8 s |J| / \omega$ around $\midp (J)$.
We claim that $\mu \in \tilde{J}$, where $\mu$ is the mean of $\normalpdf^{r, J}_{\tilde{\mu}, \tilde{\tau}}$.

Assume that $\midp(J) \leq \mu$. Let $J_0 = J$ and inductively, for $i < 4 s / \omega$, let $J_i$ be the interval with left endpoint at the right endpoint of $J_{i - 1}$ and with length $|J|$. 
That is, the $J_i$ consist of $4 s / \omega$ consecutive, non-intersecting copies of $J$ starting at $J$ and going upwards on the number line (for simplicity of exposition we assume that $4 s / \omega$ is an integer).
Let $J^\dagger = \cup_{i = 0}^{(4 s / \omega) - 1} J_i$. 
We claim that $\mu \in J^\dagger$.
Suppose not. 
This means that $\mu$ is strictly greater than any point in any $J_i$.
In particular, this implies that for all $i$,
\begin{align*} 
\int_{J_i} \normalpdf^{r, J}_{\tilde{\mu}, \tilde{\tau}} \diff x &\geq \int_{J_{i - 1}} \normalpdf^{r, J}_{\tilde{\mu}, \tilde{\tau}} \diff x \\
&\geq \int_{J_{0}} \normalpdf^{r, J}_{\tilde{\mu}, \tilde{\tau}} \diff x \\
&\geq \frac{\omega}{8 s} \;. 
\end{align*}
But then this would imply that 
\[\int_{J^\dagger}\normalpdf^{r, J}_{\tilde{\mu}, \tilde{\tau}} \diff x = \sum_{i = 0}^{(4 s / \omega) - 1}  \int_{J_i}\normalpdf^{r, J}_{\tilde{\mu}, \tilde{\tau}} \diff x \geq \frac{1}{2} \;.\]
Notice that $J^\dagger$ is itself an interval.
But any interval containing at least $1/2$ of the weight of any Gaussian must contain its mean, which we assumed did not happen.
Thus we conclude that $\mu \in J^\dagger$.
Moreover, $J^\dagger \subseteq \tilde{J}$, so $\mu \in \tilde{J}$, as claimed.
If $\midp(J) \geq \mu$ then apply the symmetric argument with $J_i$ which are decreasing on the number line instead of increasing.

We have thus shown that $\mu \in \tilde{J}$.
It is a straightforward calculation to show that this implies that $\tilde{\mu} \in [-\frac{4 s \tau}{\omega}, \frac{4 s \tau}{\omega}]$.
By the above, we know that $\tau \leq \phi / 2$ and thus $\tilde{\mu} \in [-\frac{2 s \phi}{\omega}, \frac{2 s \phi}{\omega}]$, as claimed.

\end{proof}

\begin{lemma}
\label{lem:perturbgeneral}
For any interval $J$, and $\tilde{\mu}_1, \tilde{\tau}_1, \tilde{\mu}_2, \tilde{\tau}_2$ so that $|\tilde{\tau}_i| \leq 2 \phi$ for $i \in \{1, 2\}$ and $|\tilde{\mu_1} - \tilde{\mu}_2| + | \tilde{\tau_1} - \tilde{\tau}_2 | \leq O((\epsilon / (\phi k ))^2)$, we have
\[ \| \normalpdf^{r, J}_{\tilde{\mu}_1, \tilde{\tau}_1}(x) - \normalpdf^{r, J}_{\tilde{\mu}_2, \tilde{\tau}_2} (x) \|_1 \leq \epsilon \; . \]
\end{lemma}
\begin{proof}
This follows by a change of variables and Lemma \ref{lem:perturbstandard}.
\end{proof}

Moreover, this rescaled parametrization naturally lends itself to approximation by a piecewise polynomial, namely, replace the standard normal Gaussian density function in Equation (\ref{eq:rescaled}) with $\tilde{P}_\epsilon$.
This is the piecewise polynomial that we will use to represent each individual component in the Gaussian mixture.

\subsection{A parametrization scheme for \texorpdfstring{$k$}{k}-GMMs}
In the rest of this section, our parametrization will often be of the form described above.
To distinguish this from the previous notation, for any $\theta \in \paramset_k$, and any set of $k$ intervals $J_1, \ldots, J_k$, we will let $\theta^r \in \paramset_k$ denote the rescaled parameters so that if the $i$-th component in the mixture represented by $\theta$ has parameters $w_i, \mu_i, \tau_i$, then the $i$-th component in the mixture represented by $\theta^r$ has parameters $w_i, \tilde{\mu}_i, \tilde{\tau}_i$ so that $\normalpdf_{\mu_i, \tau_i} = \normalpdf^{r, J_i}_{\tilde{\mu}_i, \tilde{\tau}_i}$.
Notice that the transformation between the original and the rescaled parameters is a linear transformation, and thus trivial to compute and to invert.

The final difficulty is that we do not know how many mixture components have associated interval $J$ for $J \in \setI$.
To deal with this, our algorithm simply iterates over all possible allocations of the mixture components to intervals and returns the best one.
There are $O(k)$ possible associated intervals $J$ and $k$ different components, so there are at most $k^{O(k)}$ different possible allocations.
In this section, we will see how our parametrization works when we fix an allocation of the mixture components.

More formally, let $\mathcal{A}$ be the set of functions $v : [s] \to \N$ so that $\sum_{\ell = 1}^s v (\ell) = k$.
These will represent the number of components ``allocated'' to exist on the scale of each $J_\ell$.
For any $v \in \mathcal{A}$, define $\setI_v$ to be the set of $I_\ell \in \setI$ so that $v(\ell) \neq 0$.

Fix $\theta^r \in \paramset_k$ and $v \in \mathcal{A}$.
Decompose $\theta^r$ into $(\theta^r_1, \ldots, \theta^r_{s})$, where $\theta^r_\ell$ contains the rescaled parameters with respect to $J_\ell$ for the $v (\ell)$ components allocated to interval $J_\ell$ (note that $v(\ell)$ may be $0$ in which case $\theta_\ell$ is the empty set, i.e., corresponds to the parameters for no components).
For any $1 \leq \ell \leq s$, let
\[\mixpdf^r_{\ell, \theta^r_\ell }(x) = \sum_i w_i  \frac{\tilde{\tau}_i}{|I_\ell| / 2} \normalpdf \left( \tilde{\tau}_i \cdot \frac{x - \midp (I_\ell)}{|I_\ell| / 2} - \tilde{\mu}_i \right)  \;  ,\]
where $i$ ranges over the components that $\theta_j$ corresponds to, and define
$\mixpdf^r_{\theta^r, v} (x) = \sum_{\ell = 1}^s \mixpdf^r_{\ell, \theta^r_\ell} (x)$.
Similarly, define
\[P^r_{\epsilon, \ell, \theta^r_\ell }(x) = \sum_i w_i  \frac{\tilde{\tau}_i}{|I_\ell| / 2} \tilde{P}_\epsilon \left( \tilde{\tau}_i \cdot \frac{x - \midp (I_\ell)}{|J_\ell| / 2} - \tilde{\mu}_i \right)  \;  ,\]
 and define
$P^r_{\epsilon, \theta^r, v} (x) = \sum_{\ell = 1}^s P^r_{\epsilon, \ell, \theta^r_\ell} (x)$.
Finally, for any $v$, define $\setP^r_{\epsilon, v}$ to be the set of all such $P^r_{\epsilon, \theta, v}$.

We have:
\begin{lemma}
\label{lem:goodapproxrescaled}
For any $\theta^r \in \paramset_k$, we have
\[\| \mixpdf^r_{\theta^r, v} - P^r_{\epsilon, \theta^r, v} \|_1 \leq \epsilon \; .\]
\end{lemma}
\noindent This follows from roughly the same argument as in the proof of Lemma \ref{lem:goodapprox}, and so we omit the proof.

We now finally have all the necessary language and tools to prove the following theorem:
\begin{corollary}
\label{thm:main-general}
Fix $2 \geq \epsilon > 0$. There is some allocation $v \in \mathcal{A}$ and a set of parameters $\theta^r \in \paramset_k$ so that $\tilde{\mu}_i \in  [-\frac{2 s \phi}{\omega} , \frac{2 s \phi}{\omega}]$, $1 / (8s) \leq \tilde{\tau}_i \leq \phi / 2$, and $w_\ell \geq \epsilon / (2k)$ for all $i$.
Moreover,
\[\| f - \mixpdf^r_{\theta^r, v} \|_1 \leq 19 \cdot \OPT_k + O(\epsilon) \; .\]
\end{corollary}
\begin{proof}
Let $\theta^* \in \paramset_k$ be so that $\| f - \mixpdf_{\theta^*} \|_1 = \OPT_k$, and let $\normalpdf_\ell^*$ denote its $\ell$-th component with parameters $w_i^\ast$, $\mu_i^\ast,$ and $\tau_i^\ast$.
Decompose $[k]$ into $S_\good (\theta^*), S_\bad (\theta^*)$ as in Lemma \ref{lem:cluster-struct}.

By the guarantees of the density estimation algorithm, we know that
\[ \left\| \sum_{\ell} w_\ell^* \normalpdf_{\mu_\ell^\ast, \tau_\ell^\ast} - \pdens \right\|_1 \leq 5 \OPT_k + \epsilon \; .\]
By Lemma \ref{lem:cluster-struct}, this implies that
\[ 5 \OPT_k + \epsilon \geq \left\| \sum_{\ell \in S_\good (\theta^*)} w_\ell^* \normalpdf_{\mu_\ell^\ast, \tau_\ell^\ast} - \pdens \right\|_1 +  \frac{1}{2} \sum_{\ell \in S_\bad (\theta^*)} w_\ell - 2\epsilon\; ,\]
from which we may conclude the following two inequalities:
\begin{align*}
 \left\| \sum_{\ell \in S_\good (\theta^*)} w_\ell^\ast \normalpdf_{\mu_\ell^\ast, \tau_\ell^\ast} - \pdens \right\|_1 &\leq 5 \cdot \OPT_k + 3 \epsilon, \numberthis \label{eq:goodbound} \\
 \sum_{\ell \in S_\bad (\theta^*)} w_\ell^\ast &\leq 10 \cdot \OPT_k + 6 \epsilon\; . \numberthis \label{eq:badbound}
\end{align*}
Let $\theta'$ be defined so that for all $\ell \in S_\good (\theta^\ast)$, the means and variances of the $\ell$-th component in $\theta'$ are $\mu_i^\ast$ and $\tau_i^\ast$, and so that for all $\ell \in S_\bad (\theta^\ast)$, the means of and variances of the $\ell$-th component in $\theta'$ are arbitrary but so that the underlying Gaussian is admissible.
Let the weights of the components in $\theta'$ be the same as the weights in $\theta^\ast$.

Then we have
\begin{align*}
\| M_{\theta'} - f \|_1 &= \left\| \sum_{\ell \in S_\good (\theta^*)} w_\ell^\ast \normalpdf_{\mu_\ell^\ast, \tau_\ell^\ast}  + \sum_{\ell \in S_\bad (\theta^\ast)} w_\ell^* \normalpdf_{\mu_\ell', \tau_\ell'}  - f \right\|_1  \\
&\leq \left\| \sum_{\ell \in S_\good (\theta^*)} w_\ell^\ast \normalpdf_{\mu_\ell^\ast, \tau_\ell^\ast} - f \right\|_1 + \left\| \sum_{\ell \in S_\bad (\theta^\ast)} w_\ell^\ast \normalpdf_{\mu_\ell', \tau_\ell'} \right\|_1 \\
&=  \left\| \sum_{\ell \in S_\good (\theta^*)} w_\ell^\ast \normalpdf_{\mu_\ell^\ast, \tau_\ell^\ast} - f \right\|_1 +  \sum_{\ell \in S_\bad (\theta^*)} w_\ell^\ast \\
&\leq \left\| \sum_{\ell \in S_\good (\theta^*)} w_\ell^\ast \normalpdf_{\mu_\ell^\ast, \tau_\ell^\ast} - \pdens \right\|_1 + \| f - \pdens \|_1 +  \sum_{\ell \in S_\bad (\theta^*)} w_\ell^\ast \\
&\leq 19 \cdot \OPT_k + O (\epsilon)
\end{align*}
where the last line follows from Equation (\ref{eq:goodbound}), the guarantee of the density estimation algorithm, and Equation (\ref{eq:badbound}). 

For each $\ell \in [k]$, let $J_\ell \in \setI$ denote the interval so that the $\ell$-th component of $\theta'$ is admissible with respect to $J_\ell$ Let $\theta^r$ be the rescaling of $\theta'$ with respect to $J_1, \ldots, J_\ell$. Then by Lemma \ref{lem:rescaled-robust}, $\theta^r$ satisfies that $\tilde{\mu}_i \in [-\frac{2 s \phi}{\omega} , \frac{2 s \phi}{\omega}]$ and $\sqrt{2\pi}\cdot \omega  / (16s) \leq  \tilde{\tau}_i \leq \phi / 2$ for all $i$.
Let $v \in \mathcal{A}$ be chosen so that $v(i)$ is the number of times that $I_i$ appears in the sequence $J_1, \ldots, J_k$. Then $\mixpdf_{\theta'}$ and $v$ satisfies all conditions in the lemma, except possibly that the weights may be too small.

Thus, let $\theta$ be the set of parameters whose means and precisions are exactly those of $\theta'$, but for which the weight of the  $\ell$-th component is defined to be $w_\ell = \max (\epsilon / (2k), w^*_\ell)$ for all $1 \leq \ell \leq k - 1$ and $w_k = 1 - \sum_{\ell = 1}^{k - 1} w_\ell$.
It is easy to see that $\theta \in \paramset_k$; moreover, $\| \mixpdf_\theta - \mixpdf_{\theta'} \|_1 \leq \epsilon$.
Then it is easy to see that $\theta$ and $v$ together satisfy all the conditions of the lemma.
\end{proof}

\subsection{The full algorithm}
At this point, we are finally ready to describe our algorithm \textsc{LearnGMM} which agnostically and properly learns an arbitrary mixture of $k$ Gaussians. 
Informally, our algorithm proceeds as follows. 
First, using \textsc{Estimate-Density}, we learn a $\pdens'$ that with high probability is $\epsilon$-close to the underlying distribution $f$ in $L_1$-distance. 
Then, as before, we may rescale the entire problem so that the density estimate is supported on $[-1, 1]$.
Call the rescaled density estimate $\pdens$.

As before, it suffices to find a $k$-GMM that is close to $\pdens$ in $\Ak$-distance, for $K = 4k - 1$.
The following is a direct analog of Lemma \ref{lem:Akapprox}.
We omit its proof because its proof is almost identical to that of Lemma \ref{lem:Akapprox}.
\begin{lemma}
\label{lem:generalAkapprox}
Let $\eps > 0, v \in \mathcal{A}$, $k \geq 2$, $\theta^r \in \paramset_k$, and $K = 4(k-1) + 1$.
Then we have 
\[
  0 \; \leq \; \| \pdens - \polypdf^r_{\epsilon, \theta^r, v} \|_{1} - \| \pdens - \polypdf^r_{\epsilon, \theta^r, v} \|_\Ak \; \leq \; 8 \cdot \OPT_k + O(\epsilon) \; .\]
\end{lemma}

Our algorithm enumerates over all $v \in \mathcal{A}$ and for each $v$ finds a $\theta^r$ approximately minimizing
\[\| \pdens - P^r_{\epsilon, \theta^r, v} \|_{\Ak} \; .\]
Using the same binary search technique as before, we can transform this problem into $\log 1 / \epsilon$ feasibility problems of the form
\begin{equation}
\label{eq:ak-general}
\| \pdens - P^r_{\epsilon, \theta^r, v} \|_{\Ak} < \eta\; .
\end{equation}
Fix $v \in \mathcal{A}$, and recall $\setP^r_{\epsilon, v}$ is the set of all polynomials of the form $P^r_{\epsilon, \theta^r, v}$.
Let $\paramset_k^\text{valid}$ denote the set of $\theta^r \in \paramset_k$ so that $\tilde{\mu}_i \in[-\frac{2 s \phi}{\omega} , \frac{2 s \phi}{\omega}]$, $\sqrt{2 \pi} \omega / (8 s) \leq \tilde{\tau} \leq \phi / 2$, and $w_i \geq \epsilon / (2k)$, for all $i$. 
For any $\theta^r \in \paramset_k^\text{valid}$, canonically identify it with $\polypdf^r_{\epsilon, \theta^r, v}$.
By almost exactly the same arguments used in Section \ref{sec:polyprogramGMM}, it follows that the class $\setP^r_{\epsilon, v}$, where $\theta \in \paramset_k^\text{valid}$, satisfies the conditions in Section \ref{subsec:polyprogram}, and that the system of polynomial equations $S_{K, \pdens, \setP^r_{\epsilon, v} } (\nu)$ has two levels of quantification (each with $O(k)$ bound variables), has $k^{O(k)}$ polynomial constraints, and has maximum degree $O(\log (1 / \epsilon))$.
Thus, we have
\begin{corollary}
\label{cor:polyprogcorrect}
For any fixed $\epsilon, \nu$, and for $K = 4k - 1$, we have that $S_{K, \pdens, \setP^r_{\epsilon, \nu}, \paramset^\text{valid}_k} (\nu)$ encodes Equation (\ref{eq:ak-general}) ranging over $\theta \in \paramset^\text{valid}_k$. 
Moreover, for all $\gamma, \lambda \geq 0$, $\textsc{Solve-Poly-Program} (S_{K, \pdens, \setP^r_{\epsilon, \nu}, \paramset^\text{valid}_k} (\nu), \lambda, \gamma)$ runs in time 
\[(k \log (1 / \epsilon))^{O(k^4)} \log \log (3 + \frac{\gamma}{\lambda}) \; . \]
\end{corollary}
 
For each $v$, our algorithm then performs a binary search over $\eta$ to find the smallest (up to constant factors) $\eta$ so that Equation (\ref{eq:ak-general}) is satisfiable for this $v$, and records both $\eta_v$, the smallest $\eta$ for which Equation (\ref{eq:ak-general}) is satisfiable for this $v$, and the output $\theta_v$ of the system of polynomial inequalities for this choice of $\eta.$ 
We then return $\theta_{v'}$ so that the $\eta_{v'}$ is minimal over all $v \in \mathcal{A}$.
The pseudocode for \textsc{LearnGMM} is in Algorithm \ref{alg:genalg}.

\begin{algorithm}[!htb]
\begin{algorithmic}[1]
\Function{LearnGMM}{$k, \epsilon, \delta$}
\LineComment{Density estimation. Only this step draws samples.}
\State $\pdens' \gets \textsc{Estimate-Density}(k, \epsilon, \delta)$
\vspace{.3cm}
\LineComment{Rescaling}
\LineComment $\pdens$ is a rescaled and shifted version of $\pdens'$ such that the support of $\pdens$ is $[-1, 1]$.
\State Let $\pdens(x) \ed \pdens'\left(\frac{2(x - \alpha)}{\beta - \alpha} - 1 \right)$
\vspace{.3cm}
\LineComment{Fitting shape-restricted polynomials}
\For{$v \in \mathcal{A}$}
	\State $\eta_v, \theta^r_v \gets \textsc{FindFitGivenAllocation}(\pdens, v)$
\EndFor
\State Let $\theta$ so that $\theta^r = \theta^r_{v'}$ so that $\eta_{v'}$ is minimal over all $\eta_v$ (breaking ties arbitrarily).
\LineComment{Round weights back to be on the simplex} \label{line:roundtosimplex}
\For{$i = 1, \ldots, k - 1$}
	\State{$w_i \gets w_i - \epsilon / 2k$ (This guarantees that $\sum_{i = 1}^{k - 1} w_i \leq 1$; see analysis for details)}
	\State{If $w_i > 1$, set $w_i = 1$}
\EndFor
\State{$w_k \gets 1 - \sum_{i = 1}^{k - 1} w_i$}
\LineComment{Undo the scaling}
\State $w_i' \gets w_i$
\State $\mu_i' \gets \frac{(\mu_i + 1)(\beta - \alpha)}{2} + \alpha$
\State $\tau_i' \gets \frac{\tau_i}{\beta - \alpha}$
\State \textbf{return} $\theta'$
\vspace{.3cm}
\EndFunction

\Function{FindFitGivenAllocation}{$\pdens, v$}
\State $\nu \gets \epsilon$
\State Let $C_1$ be a universal constant sufficiently small.
\State Let $\lambda \gets \min (C_1 (\epsilon / (\phi k ))^2, 1 / 16s, \epsilon / (4k))$ 
\LineComment{This choice of precision provides robustness as needed by Lemma \ref{lem:perturbgeneral}, and also ensures that all the weights and precisions returned must be non-negative.}
\State Let $\psi \gets 6 k s \phi / \omega + 3 k \phi / 2 + 1$
\LineComment{By Corollary \ref{thm:main-general}, this is a bound on how large any solution of the polynomial program can be.}
\State $\theta^r \gets \textsc{Solve-Poly-System}(S_{\pdens, \setP^r_{\epsilon, v}, \paramset^\text{valid}_k} (\nu), \lambda, \psi)$
\While{$\theta^r$ is ``NO-SOLUTION''}
	\State $\nu \gets 2 \cdot \nu$
  \State $\theta^r \gets \textsc{Solve-Poly-System}(S_{\pdens, \setP^r_{\epsilon, v}, \paramset^\text{valid}_k} (\nu), \lambda, \psi)$
\EndWhile
\State \textbf{return} $\theta^r, \nu$
\EndFunction
\end{algorithmic}
\caption{Algorithm for proper learning an arbitrary mixture of $k$ Gaussians.}
\label{alg:genalg}
\end{algorithm}

The following theorem is our main technical contribution:
\begin{theorem}
\textsc{LearnGMM}$(k, \epsilon, \delta)$  takes $\Otilde ((k + \log 1/\delta) / \epsilon^2)$ samples from the unknown distribution with density $f$, runs in time
\[
\left( k \log\frac{1}{\eps} \right)^{O(k^4)} + \Otilde\left( \frac{k}{\eps^2} \right) \; ,
\]
and with probability $1 - \delta$ returns a set of parameters $\theta \in \paramset_k$ so that $\| f - \mixpdf_\theta \|_1 \leq 58 \cdot \OPT + \epsilon$.
\end{theorem}
\begin{proof}
The sample complexity follows simply because $\textsc{Estimate-Density}$ draws $\Otilde ((k + \log 1/\delta) / \epsilon^2)$ samples, and these are the only samples we ever use.
The running time bound follows because $|\mathcal{A}| = k^{O(k)}$ and from Corollary \ref{cor:polyprogcorrect}.
Thus it suffices to prove correctness.

Let $\theta$ be the parameters returned by the algorithm.
It was found in some iteration for some $v \in \mathcal{A}$.
Let $v^*, \theta^*$ be those which are guaranteed by Corollary \ref{thm:main-general}.
We have
\[
\| \pdens - P^r_{\epsilon, \theta^*, v^*} \|_{\Ak} \leq \| \pdens - f  \|_1 + \| f - \mixpdf^r_{\theta^*, v^*} \|_1 + \| \mixpdf^r_{\theta^*, v^*} - P^r_{\epsilon, \theta^*, v^*} \|_1 \leq 23 \cdot \OPT_k + O( \epsilon ) \; .
\]
By the above inequalities, the system of polynomial equations is feasible for $\eta \leq 46 \cdot \OPT_k + O( \epsilon )$ in the iteration corresponding to $v^*$ (Corollary \ref{thm:main-general} guarantees that the parameters $\theta^\ast$ are sufficiently bounded).
Hence, for some $\eta_{v^\ast} \leq \eta$, the algorithm finds some $\theta'$ so that there is some $\theta''$ so that $\| \theta' - \theta'' \|_2 \leq C_1 (\epsilon / (\phi k ))^2$, which satisfies $S_{\pdens, \setP^r_{\epsilon, v^*}, \paramset^\text{valid}_k} (\nu_{v^\ast}).$ 

Let $\theta_1$ be the set of parameters computed by the algorithm before rounding the weights back to the simplex (i.e.\ at Line \ref{line:roundtosimplex}).
By our choice of precision in solving the polynomial program, (i.e. by our choice of $\lambda$ on Line 24 of Algorithm \ref{alg:genalg}), we know that the precisions of the returned mixture are non-negative (so each component is a valid Gaussian).
It was found in an iteration corresponding to some $v \in \mathcal{A}$, and there is some $\eta_v \leq \eta_{v^*} \leq 46 \cdot \OPT_k + O(\epsilon)$ and some $\theta_1'$ satisfying the system of polynomial equalities for $v$ and $\eta_v$, so that $\| \theta_1 - \theta_1' \|_2 \leq C_1 (\epsilon / (\phi k))^2$.
Let $\theta$ be the set of rescaled parameters obtained after rounding the weights of $\theta_1$ back to the simplex.
It is straightforward to check that $\theta \in \paramset_k$, and moreover, $\| \mixpdf^r_{\theta, v} - \mixpdf^r_{\theta_1', v} \|_1 \leq 2 \epsilon$, and so $\| P^r_{\epsilon, \theta, v} - P^r_{\epsilon, \theta_1', v} \|_1 \leq O(\epsilon)$.

We therefore have
\begin{align*}
\| f - M_\theta \|_1 &\leq \| f - \pdens \|_1 + \| \pdens - P^r_{\epsilon, \theta, v} \|_1 + \| P^r_{\epsilon, \theta, v} - \mixpdf^r_{\epsilon, \theta, v} \|_1 \\
&\stackrel{(a)}{\leq} 4 \cdot \OPT + \epsilon + 8 \cdot \OPT + O(\epsilon) + \| \pdens - P^r_{\epsilon, \theta, v} \|_\Ak + \epsilon \\
&\stackrel{(b)}{\leq} 12 \cdot \OPT + O(\epsilon) + \| \pdens - P^r_{\epsilon, \theta_1', v} \|_\Ak \; \\
&\stackrel{(c)}{\leq} 58 \cdot \OPT + O(\epsilon) \;,
\end{align*}
where (a) follows from Lemmas \ref{lem:generalAkapprox} and \ref{lem:goodapproxrescaled}, (b) follows from the arguments above, and (c) follows since $\theta_1'$ satisfies the system of polynomial inequalities for $\eta_v \leq 46 \cdot \OPT_k + O(\epsilon)$.

As a final step, we choose an internal $\eps'$ in our algorithm so that the $O(\eps')$ in the above guarantee becomes bounded by $\eps$.
This proves the desired approximation guarantee and completes the proof.
\end{proof}

%% file: general.tex

\section{Further classes of distributions}
\label{sec:general}

In this section, we briefly show how to use our algorithm to properly learn other parametric classes of univariate distributions.

Let $\setC$ be a class of parametric distributions on the real line, parametrized by $\theta \in S$ for $S \subseteq \R^u$.
For each $\theta$, let $F_\theta \in \setC$ denote the pdf of the distribution parametrized by $\theta$ in $\setC$.
To apply our algorithm in this setting, it suffices to show the following:
\begin{enumerate}
\item \emph{(Simplicity of $\setC$}) For any $\theta_1$ and $\theta_2$, the function $F_{\theta_1} -F_{\theta_2}$ has at most $K$ zero crossings. In fact it also suffices if any two such functions have ``essentially'' $K$ zero crossings.
\item \emph{(Simplicity of $S$)} $S$ is a semi-algebraic set.
\item \emph{(Representation as a piecewise polynomial)} For each $\theta \in S$ and any $\epsilon > 0$, there is a a piecewise polynomial $P_{\epsilon, \theta}$ so that $\| P_{\epsilon, \theta} - F_\theta \|_1 \leq \epsilon$.
Moreover, the map $(x, \theta) \mapsto P_{\epsilon, \theta} (x)$ is jointly polynomial in $x$ and $\theta$ at any point so that $x$ is not at a breakpoint of $P_{\epsilon, \theta}$.
Finally, the breakpoints of $P_{\epsilon, \theta}$ also depend polynomially on $\theta$.
\item \emph{(Robustness of the Parametrization)} There is some robust parametrization so that  we may assume that all ``plausible candidate'' parameters are $\leq 2^{\poly (1 / \epsilon)}$, and moreover, if $\| \theta_1 - \theta_2 \| \leq 2^{-\poly (1 / \epsilon)}$, then $\| F_{\theta_1} - F_{\theta_2} \| \leq \epsilon$.
\end{enumerate}
Assuming $\setC$ satisfies these conditions, our techniques immediately apply.
In this paper, we do not attempt to catalog classes of distributions which satisfy these properties.
However, we believe such classes are often natural and interesting.
We give evidence for this below, where we show that our framework produces proper and agnostic learning algorithms for mixtures of two more types of simple distributions.
The resulting algorithms are both sample optimal (up to log factors) and have nearly-linear running time.

\subsection{Learning mixtures of simple distribution}
As a brief demonstration of the generality of our technique,
we show that our techniques give proper and agnostic learning algorithms for mixtures of $k$ exponential distributions and Laplace distributions (in addition to mixtures of $k$ Gaussians) which are \emph{nearly-sample optimal}, and run in time which is \emph{nearly-linear} in the number of samples drawn, for any constant $k$.

We now sketch a proof of correctness for both classes mentioned above.
In general, the robustness condition is arguably the most difficult to verify of the four conditions required. 
However, it can be verified that for mixtures of simple distributions with reasonable smoothness conditions the appropriate modification of the parametrization we developed in Section \ref{sec:generalalg} will suffice.
Thus, for the classes of distributions mentioned, it suffices to demonstrate that they satisfy conditions (1) to (3).

\paragraph{Condition 1:} It follows from the work of \cite{Tossavainen06} that the difference of $k$ exponential distributions or $k$ Laplace distributions has at most $2k$ zero crossings.
\paragraph{Condition 2:} This holds trivially for the class of mixtures of exponential distributions. 
We need a bit of care to demonstrate this condition for Laplace distributions since a Laplace distribution with parameters $\mu, b$ has the form
\[\frac{1}{2b} e^{-|x - \mu| / b}\]
and thus the Taylor series is not a polynomial in $x$ or the parameters.
However, we may sidestep this issue by simply introducing a variable $y$ in the polynomial program which is defined to be $y = |x - \mu|$.

\paragraph{Condition 3:} It can easily be shown that a truncated degree $O(\log 1 / \epsilon)$ Taylor expansion (as of the form we use for learning $k$-GMMs) suffices to approximate a single exponential or Laplace distribution, and hence a $O(k)$-piecewise degree $O(\log 1 / \epsilon)$ polynomial suffices to approximate a mixture of $k$ exponential or Laplace distributions up to $L_1$-distance $\eps$.

Thus for both of these classes, the sample complexity of our algorithm is $\Otilde (k / \epsilon^2)$, and its running time is
\[
\left(k \log \frac{1}{\epsilon} \right)^{O(k^4)} + \Otilde \left( \frac{k}{\epsilon^2} \right) \; ,
\]
similar to the algorithm for learning $k$-GMMs.
As for $k$-GMMs, this sample complexity is nearly optimal, and the running time is nearly-linear in the number of samples drawn, if $k$ is constant.